\documentclass[a4paper]{article}
\usepackage{amsmath,amsthm,amssymb,amsfonts,mathrsfs,eucal,amsthm}
\usepackage{dsfont,yfonts,calligra}
\usepackage{graphicx,hyperref}
\usepackage{pxfonts,upgreek}
\usepackage{lmodern}

\DeclareMathAlphabet{\mathpzc}{OT1}{pzc}{m}{it}
\DeclareMathAlphabet{\mathcalligra}{T1}{calligra}{m}{n}

\def\Tr{\operatorname{Tr}} \def\>{\rangle} \def\<{\langle}
 \def\id{\mathsf{id}}
\def\test#1{\mathfrak{#1}} \def\outx{\mathcal{X}}
\def\outy{\mathcal{Y}} \def\mE{\mathcal{E}} \def\mL{\mathcal{L}}
\def\sH{\mathcal{H}} \def\sK{\mathcal{K}} \def\povm#1{\mathrm{#1}}
\def\sS{\mathfrak{S}}  
\def\bound{\textswab{L}} \def\insym{\mathrm{in}}
\def\outsym{\mathrm{out}}

 \def\exp#1{\mathscr{#1}}
\def\qexp#1{\boldsymbol{\mathsf{#1}}} \def\dec#1{\mathscr{D}}
\def\s{\mathsf{s}} \def\openone{\mathds{1}}
\def\msucc{\succ_{\!\mathrm{m}}}

\renewcommand{\qedsymbol}{\nobreak \ifvmode \relax \else
  \ifdim \lastskip<1.5em \hskip-\lastskip \hskip1.5em plus0em
  minus0.5em \fi \nobreak \vrule height0.75em width0.5em
  depth0.25em\fi}

\renewcommand{\ge}{\geqslant}
\renewcommand{\le}{\leqslant}

\newtheorem{theorem}{Theorem}

\newtheorem{lemma}{Lemma}
\newtheorem{proposition}{Proposition}
\newtheorem{postulate}{Postulate}

\newtheorem*{proposition0}{Proposition 1}

\theoremstyle{remark}
\newtheorem{remark}{Remark}

\theoremstyle{definition}
\newtheorem{definition}{Definition}

\begin{document}



\title{Comparison of quantum statistical models:\\equivalent conditions
  for sufficiency}



\author{{\sc Francesco Buscemi}\footnote{buscemi@iar.nagoya-u.ac.jp}\\
  \small{\em Institute for Advanced Research,
    University of Nagoya}\\
  \small{\em Chikusa-ku, Nagoya 464-8601, Japan.}}





\date{}

\maketitle

\begin{abstract}

  A family of probability distributions (i.e. a statistical model) is
  said to be sufficient for another, if there exists a transition
  matrix transforming the probability distributions in the former to
  the probability distributions in the latter. The
  Blackwell-Sherman-Stein (BSS) theorem provides necessary and
  sufficient conditions for one statistical model to be sufficient for
  another, by comparing their information values in statistical
  decision problems. In this paper we extend the BSS theorem to
  quantum statistical decision theory, where statistical models are
  replaced by families of density matrices defined on
  finite-dimensional Hilbert spaces, and transition matrices are
  replaced by completely positive, trace-preserving maps
  (i.e. coarse-grainings). The framework we propose is suitable for
  unifying results that previously were independent, like the BSS
  theorem for classical statistical models and its analogue for pairs
  of bipartite quantum states, recently proved by Shmaya. An important
  role in this paper is played by \emph{statistical morphisms},
  namely, affine maps whose definition generalizes that of
  coarse-grainings given by Petz and induces a corresponding criterion
  for statistical sufficiency that is weaker, and hence easier to be
  characterized, than Petz's.
\end{abstract}




\section{Introduction}

The task in which an experimenter tries to learn about the true value
of an unknown parameter by observing a random variable whose
distribution depends on such a value, is generally called a
statistical estimation task or a statistical decision problem. The
mathematical structure used to describe such a situation is called
\emph{statistical model}, i.e. a family of probability distributions
(or, more generally, measures) indexed by a parameter set, which
represents the unknown parameter one wants to learn about in the
estimation process.

An important subject in classical statistics is the comparison of
statistical models in terms of their ``information value'' in
statistical decision problems~\cite{black49, wald, black1, sherman,
  stein, blackwell, lecam1, torgersen, lecam, goel-ginebra}. Within
this area, one of the most important results has been proved by
Blackwell, Sherman, and Stein~\cite{black1, sherman, stein,
  blackwell}: the theorem states equivalent conditions for one
statistical model being \emph{more informative} than another. More
explicitly, the Blackwell-Sherman-Stein (from now on, BSS) theorem
proves that one statistical model carries more information than
another if and only if the former is \emph{sufficient} for the latter,
namely, if and only if there exists a transition matrix (a Markov
kernel) mapping the probability distributions (measures) in the former
to the probability distributions (measures) in the latter.

In quantum statistical decision theory~\cite{holevo,ozawa}, where
statistical models are replaced by families of non-commuting density
operators (i.e. quantum statistical models), the notion of sufficiency
has been introduced and developed by Petz~\cite{petz,petz2}, by
replacing Markov kernels with completely positive (or, at least,
two-positive) trace-preserving maps, i.e. coarse-grainings. However,
the idea of applying to the quantum case concepts from the theory of
comparison of statistical models \emph{\`a la} BSS, like e.g. the
concept of information value, has not been explicitly pursued until
recently, in a work by Shmaya~\cite{shmaya}: there, partial ordering
relations between pairs of bipartite quantum states (analogous, in a
way, to the partial ordering relations used in the BSS theorem) are
introduced, and an equivalence relation between such partial orderings
is established. Subsequently, in~\cite{chefles}, Chefles reformulated
Shmaya's result for the comparison of pairs of quantum
channels. However, the equivalence relations proved in~\cite{shmaya}
and~\cite{chefles} neither imply any criterion for the comparison of
quantum statistical models, nor are they more general than the BSS
theorem, with which they are, in fact, logically
\emph{unrelated}. This is due to the fact that both Shmaya and Chefles
need, in their proofs, quantum entanglement: as such, their results
are \emph{purely quantum} and cannot be compared with the case of
classical statistics, where quantum entanglement is not available.

The aim of this paper is to bridge the gap mentioned above, by
developing a general theory for the comparison of statistical models,
which can be applied both to the classical and the quantum
(i.e. non-commutative) setting. In order to do this, it is
mathematically convenient to relax the definition of sufficiency
introduced by Petz~\cite{petz,petz2} and define a weaker notion of
sufficiency, which we call \emph{m-sufficiency}, based on the concept
of \emph{statistical morphisms}\footnote{The term ``statistical
  morphism'' has been introduced in the classical setting by Morse and
  Sacksteder~\cite{morse}. In this paper we use the same term, but in
  a non-commutative setting.}. Statistical morphisms are affine maps
satisfying the minimum requirements necessary to make them meaningful
in a statistical sense: in fact, as we will carefully argue in what
follows, even the requirement of positivity can be lifted, without
compromising the formalism. In spite of their generality, statistical
morphisms are sufficiently well-behaved, so that, in some cases, they
can be extended to completely positive coarse-grainings. This fact is
proved in two extension theorems, of crucial importance in this paper,
analogous to those proved for positive maps by Choi (Theorem~6
in~\cite{choi}) and Arveson (Proposition 1.2.2 in~\cite{arveson}).

The generality of the definition of statistical morphisms makes the
main result proved here applicable to both commutative and
non-commutative scenarios. When specialized to the classical setting,
our result provides an alternative proof of the BSS theorem, while,
in the quantum setting, an equivalent characterization of Petz's
sufficiency criterion is obtained. An intermediate, `hybrid'
quantum-classical case is also considered and completely
characterized. We are also able to recover Shmaya's result as a
special case, although here, in contrast with Refs.~\cite{shmaya}
and~\cite{chefles}, we \emph{never} need to resort to any additional
entangled resource.

The paper is organized as follows: in Section~\ref{sec:classical} we
briefly review the notions of statistical models, statistical decision
problems, and comparison of statistical models in classical
statistics. In Section~\ref{sec:defs} we introduce some basic
definitions, extending the idea of comparison of statistical models to
finite dimensional quantum systems. In Section~\ref{sec:proc}, we
introduce the notions of statistical morphisms and m-sufficiency. In
Section~\ref{sec:ext} we prove two extension theorems for statistical
morphisms. Section~\ref{sec:main} contains the main result, which is
then applied, in Section~\ref{sec:black}, in order to recover the BSS
theorem, characterize a semi-classical scenario, and obtain an
equivalent characterization of Petz's sufficiency
relation. Section~\ref{sec:shmaya} deals with the scenario originally
considered in Ref.~\cite{shmaya} and the result proved by Shmaya is
recovered without the need of any entangled auxiliary
resource. Finally, Section~\ref{sec:concl} concludes the paper with
the summary of its contents and one remark about generalized
probabilistic theories.

\section{Classical formulation}\label{sec:classical}

A (finite) \emph{statistical model} $\exp{E}$ is defined by a triple
$(\Theta,\Delta,\boldsymbol{\alpha})$, where $\Theta$ is a (finite)
parameter set $\{\theta\}_{\theta\in\Theta}$, $\Delta$ is a (finite)
sample set $\{\delta\}_{\delta\in\Delta}$, and $\boldsymbol{\alpha}$
is a family $(p_\theta;\theta\in\Theta)$ of probability distributions
$p_\theta$ on $\Delta$, i.e., $p_\theta(\delta)\ge 0$ and
$\sum_{\delta\in\Delta}p_\theta(\delta)=1$. In the following, it will
sometimes be convenient to think of each $p_\theta$ as a
$|\Delta|$-dimensional probability vector
$\vec{p}_\theta=(p^1_\theta\;,\,\cdots\;,\,p_\theta^\delta\;,\,\cdots\;,\,p^{|\Delta|}_\theta)$,
whose components are defined as $p_\theta^\delta:=p_\theta(\delta)$.

\begin{remark}\label{rem:states}
  In many relevant situations, $\Delta$ can be considered as the set
  of possible states of a physical system, so that the probability
  distribution $p_\theta$ becomes the statistical description of the
  state of the system. This point of view, which is the guiding one in
  Ref.~\cite{holevo}, will be implicitly adopted here as well.
\end{remark}

A \emph{statistical decision problem} is defined by a triple
$(\exp{E},\outx,\ell)$, where
$\exp{E}=(\Theta,\Delta,\boldsymbol{\alpha})$ is a statistical model,
$\outx$ is a (finite) decision set $\{i\}_{i\in\outx}$, and
$\ell:\Theta\times\outx\to\mathbb{R}$ is a payoff function. The
decision problem works as follows: upon the observation (or state)
$\delta\in\Delta$, occurring with probability $p_\theta(\delta)$, the
statistician performs a decision, namely, he applies a
$\outx$-decision function $u:\Delta\to\outx$, gaining a payoff (or
suffering a loss, if negative) of $\ell(\theta,i)$, depending on the
``true'' law of nature $\theta$ that determined the observed state
$\delta$. The choice of the function $u:\Delta\to\outx$ corresponds to
the experimenter's choice of a \emph{strategy}.

The deterministic $\outx$-decision function $u:\Delta\to\outx$ is
often generalized to a \emph{randomized $\outx$-decision function} (or
$\outx$-r.d.f.)  $\phi$, which is a convex combination of
$\outx$-decision functions, i.e., a function mapping each
$\delta\in\Delta$ to a probability distribution $t_\delta$ on
$\outx$. A convenient way to represent a $\outx$-r.d.f. $\phi$ is by
giving conditional probabilities $t_\phi(i|\delta)\ge 0$,
i.e. non-negative real numbers such that
$\sum_{i\in\outx}t_\phi(i|\delta)=1$, for all $\delta\in\Delta$.

Given a decision problem $(\exp{E},\outx,\ell)$, for each
$\outx$-r.d.f. $\phi$, we introduce the payoff vector
$\vec{v}(\phi;\exp{E},\outx,\ell)\in\mathds{R}^{|\Theta|}$, whose
$\theta$-th component, representing the payoff gained if the true law
of nature is $\theta$, is defined as
\begin{equation}\label{eq:vn}
  v^\theta(\phi;\exp{E},\outx,\ell):=
  \sum_{i\in\outx}\ell(\theta,i)\sum_{\delta\in\Delta}t_\phi(i|\delta)p_{\theta}(\delta).
\end{equation}
Then, the following set
\begin{equation}
  \mathcal{C}(\exp{E},\outx,\ell):=\left\{\vec{v}(\phi;\exp{E},\outx,\ell)\left|\textrm{$\phi$
        is a $\outx$-r.d.f. on $\Delta$}\right.\right\}
\end{equation}
forms a (closed and bounded) convex subset of $\mathds{R}^{|\Theta|}$,
since it inherits the convex structure from the set of randomized
decision functions.

Let now $\exp{F}=(\Theta,\Delta',\boldsymbol{\beta})$ be another
statistical model, with the same parameter set $\Theta$ as for
$\exp{E}$, but with a different sample set $\Delta'$ and a different
family of probability distributions on $\Delta'$,
$\boldsymbol{\beta}=(q_\theta;\theta\in\Theta)$. Also for $\exp{F}$,
we can define, for each decision set $\outx$ and each payoff function
$\ell:\Theta\times\outx\to\mathbb{R}$, the convex set of achievable
payoff vectors as
\begin{equation}
  \mathcal{C}(\exp{F},\outx,\ell):=\left\{\vec{v}(\phi';\exp{F},\outx,\ell)\left|\textrm{$\phi'$
        is a $\outx$-r.d.f. on $\Delta'$}\right.\right\}.
\end{equation}

In classical statistics, the following partial ordering between
statistical models with the same parameter set $\Theta$ is introduced (see,
e.~g., Ref.~\cite{blackwell}):

\begin{definition}[Information Ordering]\label{def:class_order}
  The statistical model $\exp{E}=(\Theta,\Delta,\boldsymbol{\alpha})$
  is said to be \emph{always more informative than}
  $\exp{F}=(\Theta,\Delta',\boldsymbol{\beta})$, in formula
  $\exp{E}\supset\exp{F}$, if and only if, for every finite set of
  decisions $\outx$ and every payoff function
  $\ell:\Theta\times\outx\to\mathbb{R}$,
  $\mathcal{C}(\exp{E},\outx,\ell)\supseteq
  \mathcal{C}(\exp{F},\outx,\ell)$.
\end{definition}

In other words, $\exp{E}$ is said to be more informative than
$\exp{F}$ if every payoff vector attainable in the problem
$(\exp{F},\outx,\ell)$ is also attainable in the problem
$(\exp{E},\outx,\ell)$. The definition of information ordering between
statistical models can be simplified as follows. Given a statistical
model $\exp{E}=(\Theta,\Delta,\boldsymbol{\alpha})$, for every
decision problem $(\exp{E},\outx,\ell)$ and every
$\outx$-r.d.f. $\phi$, we define
\begin{equation}\label{eq:str-exp-pay}
\mathsf{s}(\exp{E},\outx,\ell,\phi):=\frac 1{|\Theta|}\sum_{\theta\in\Theta}\sum_{i\in\outx}\ell(\theta,i)\sum_{\delta\in\Delta}t_{\phi}(i|\delta)p_\theta(\delta).
\end{equation}
The maximum of $\mathsf{s}(\exp{E},\outx,\ell,\phi)$ over all
$\outx$-r.d.f. $\phi$ is correspondingly defined as
\begin{equation}\label{eq:max-exp-pay}
\begin{split}
  \$(\exp{E},\outx,\ell):&=\max_{\phi:\textrm{ $\outx$-r.d.f.}}\frac 1{|\Theta|}\sum_{\theta\in\Theta}\sum_{i\in\outx}\ell(\theta,i)\sum_{\delta\in\Delta}t_{\phi}(i|\delta)p_\theta(\delta)\\
  &=\max_{\phi:\textrm{ $\outx$-r.d.f.}}\frac
  1{|\Theta|}\sum_{\theta\in\Theta}v^{\theta}(\phi;\exp{E},\outx,\ell).
\end{split}
\end{equation}
In the Bayesian approach, when there is no compelling reason to treat
the sample set differently from the parameter set, it is reasonable
to interpret the factor $1/|\Theta|$ as an \emph{a priori} probability
distribution over the unknown parameter $\theta\in\Theta$. In this
framework, the function $\$(\exp{E},\outx,\ell)$ is understood
as the optimal expected payoff, and the following partial ordering
between statistical models governed by the same parameter set
$\Theta$ is introduced (see, e.~g., Ref.~\cite{goel-ginebra}):

\begin{definition}[Bayesian Information
  Ordering]\label{def:bayes_order}
  The statistical model $\exp{E}=(\Theta,\Delta,\boldsymbol{\alpha})$
  is said to be \emph{bayesianly always more informative than}
  $\exp{F}=(\Theta,\Delta',\boldsymbol{\beta})$, in formula
  $\exp{E}\supset_{\mathrm{Bayes}}\exp{F}$, if and only if, for every
  finite decision set $\outx$ and every payoff function
  $\ell:\Theta\times\outx\to\mathbb{R}$,
  $\$(\exp{E},\outx,\ell)\ge \$(\exp{F},\outx,\ell)$.
\end{definition}

In other words, $\exp{E}$ is said to be bayesianly more informative
than $\exp{F}$ if every expected payoff attainable in the problem
$(\exp{F},\outx,\ell)$ is also attainable in the problem
$(\exp{E},\outx,\ell)$. In Appendix~\ref{app:a} we report the proof of
the following basic fact:

\begin{proposition}\label{prop0}
$\exp{E}\supset \exp{F}$ if and only if
$\exp{E}\supset_{\mathrm{Bayes}}\exp{F}$.
\end{proposition}

Since the two orderings $\exp{E}\supset\exp{F}$ and
$\exp{E}\supset_{\mathrm{Bayes}}\exp{F}$ are equivalent, from now on
we will keep only the notation $\exp{E}\supset \exp{F}$, stressing the
fact that the ``Bayesian information ordering'' relation given in
Definition~\ref{def:bayes_order} does not really depend on any strong
Bayesian assumption.

Another partial ordering between statistical models with the same parameter
set $\Theta$ is relevant, and it is defined as
follows, according to~\cite{blackwell}:

\begin{definition}[Sufficiency]\label{def:class_suff}
  The statistical model $\exp{E}=(\Theta,\Delta,\boldsymbol{\alpha})$
  is said to be \emph{sufficient for}
  $\exp{F}=(\Theta,\Delta',\boldsymbol{\beta})$, in formula
  $\exp{E}\succ\exp{F}$, if and only if there exists a
  $|\Delta'|\times|\Delta|$ transition matrix $\mathrm{M}$,
  i.e. a matrix of non-negative numbers $M_{\delta',\delta}$ with
  $\sum_{\delta'\in\Delta'}M_{\delta',\delta}=1$ for all
  $\delta\in\Delta$, for which
  $\vec{q}_\theta=\mathrm{M}\vec{p}_\theta$, for all
  $\theta\in\Theta$.
\end{definition}

The Blackwell-Sherman-Stein (BSS) theorem states the following
important equivalence relation:
\begin{theorem}[BSS Theorem\cite{black1,sherman, stein}]\label{theo:blackwell-orig}
  Given two statistical models
  $\exp{E}=(\Theta,\Delta,\boldsymbol{\alpha})$ and
  $\exp{F}=(\Theta,\Delta',\boldsymbol{\beta})$, governed by the same
  parameter set $\Theta$, $\exp{E}\succ\exp{F}$ if and only if
  $\exp{E}\supset \exp{F}$.
\end{theorem}

\section{The formulation in quantum theory}\label{sec:defs}

In what follows, we only consider quantum systems defined on finite
dimensional Hilbert spaces $\sH$. We denote by $\bound(\sH)$ the set
of all linear operators (identified with their representing matrices)
acting on $\sH$, and by $\sS(\sH)$ the set of all density matrices (or
states) $\rho\in\bound(\sH)$, with $\rho\ge 0$ and $\Tr[\rho]=1$. The
identity matrix will be denoted by the symbol $\openone$, whereas the
identity map will be denoted by $\id$.

Most of the concepts used here are introduced and rigorously
formalized in Refs.~\cite{holevo} and~\cite{ozawa}. For reader's
clarity, however, we will report the definitions we need, in a
simplified fashion. According with~\cite{holevo} (see also
Remark~\ref{rem:states}), we adopt the following definition:

\begin{definition}[Quantum Statistical Models]
\label{def:exp}
A \emph{quantum statistical model} is defined by a triple
$\qexp{R}=(\Theta,\sH,\boldsymbol{\rho})$, where $\Theta$ is a
(finite) parameter set, $\sH$ is a (finite dimensional) Hilbert space,
and $\boldsymbol{\rho}=(\rho_\theta;\theta\in\Theta)$ is a family of
density matrices in $\sS(\sH)$. A quantum statistical model $\qexp{R}$
is said to be \emph{abelian} when $[\rho_\theta,\rho_{\theta'}]=0$,
for all $\theta,\theta'\in\Theta$.
\end{definition}

\begin{definition}[POVM's]\label{def:povm}
  For any (finite) decision set $\outx=\{i\}$, a
  \emph{positive-operator--valued $\outx$-measure} ($\outx$-POVM)
  $\povm{P}^\outx$ on the Hilbert space $\sH$ is a family
  $(P^i;i\in\outx)$ of operators $P^i\in\bound(\sH)$, such that
  $P^i\ge 0$ for all $i\in\outx$ and
  $\sum_{i\in\outx}P^i=\openone$. From now on, the superscript $\outx$
  will be dropped when clear from the context.
\end{definition}

\begin{definition}[Quantum Statistical Decision Problems] A
  \emph{quantum statistical decision problem} is defined by a triple
  $(\qexp{R},\outx,\ell)$, where
  $\qexp{R}=(\Theta,\sH,\boldsymbol{\rho})$ is a quantum statistical
  model, $\outx$ is a (finite) decision set $\{i\}_{i\in\outx}$, and
  $\ell:\Theta\times\outx\to\mathbb{R}$ is a payoff function. The
  choice of a \emph{strategy} for the problem
  $(\qexp{R},\outx,\ell)$ corresponds to the choice of a
  $\outx$-POVM $\povm{P}=(P^i;i\in\outx)$ on $\sH$. The corresponding
  expected payoff is computed as
  \begin{equation}
    \s_{\textrm{q}}(\qexp{R},\outx,\ell,\povm{P}):=\frac
    1{|\Theta|}\sum_{\theta\in\Theta}\sum_{i\in\outx}\ell(\theta,i)\Tr[\rho_\theta\
    P^i].
  \end{equation}
  The maximum expected payoff for the decision problem
  $(\qexp{R},\outx,\ell)$ is defined as
\begin{equation}\label{eq:q-exp-payoff}
  \$_{\textrm{q}}(\qexp{R},\outx,\ell):=\max_{\povm{P}}\frac
  1{|\Theta|}\sum_{\theta\in\Theta}\sum_{i\in\outx}\ell(\theta,i)\Tr[\rho_\theta\
  P^i].
\end{equation}
Notice the use of the subscript ``q'', for ``quantum'', to distinguish
the expressions above from their classical analogues appearing
in~(\ref{eq:str-exp-pay}) and~(\ref{eq:max-exp-pay}).
\end{definition}

Given two quantum statistical models
$\qexp{R}=(\Theta,\sH,\boldsymbol{\rho})$ and
$\qexp{S}=(\Theta,\sH',\boldsymbol{\sigma})$, governed by the same
parameter set $\Theta$, but with different Hilbert spaces $\sH$ and
$\sH'$ and different families of quantum states $\boldsymbol{\rho}
=(\rho_\theta\in\sS(\sH);\theta\in\Theta)$ and $\boldsymbol{\sigma}
=(\sigma_\theta\in\sS(\sH');\theta\in\Theta)$, the following partial
ordering is introduced:

\begin{definition}[Information Ordering]\label{def:qexp_info_order}
  A quantum statistical model
  $\qexp{R}=(\Theta,\sH,\boldsymbol{\rho})$ is said to be \emph{always
    more informative than}
  $\qexp{S}=(\Theta,\sH',\boldsymbol{\sigma})$, in formula
  $\qexp{R}\supset\qexp{S}$, if and only if, for every finite decision
  set $\outx$ and every payoff function
  $\ell:\Theta\times\outx\to\mathbb{R}$,
  $\$_{\textrm{q}}(\qexp{R},\outx,\ell)\ge
  \$_{\textrm{q}}(\qexp{S},\outx,\ell)$.
\end{definition}

In other words, $\qexp{R}$ is said to be more informative than
$\qexp{S}$ if every expected payoff attainable in the problem
$(\qexp{S},\outx,\ell)$ is also attainable in the problem
$(\qexp{R},\outx,\ell)$.

\begin{remark}
  We stress once more that the information ordering relation between
  quantum statistical models introduced above does not depend on the
  \emph{a priori} distribution on $\Theta$ used to compute the
  expected payoff~(\ref{eq:q-exp-payoff}). One could in fact adopt an
  information ordering for quantum statistical models analogous to
  that introduced in Definition~\ref{def:class_order}, and prove that
  the two ordering relations are equivalent. This is due to the fact
  that Proposition~\ref{prop0} is valid also for quantum statistical
  models.
\end{remark}

\begin{remark}[Quantum-Classical
  Correspondence]\label{rem:quantum-vs-classical}
  Given an abelian quantum statistical model
  $\qexp{R}=(\Theta,\sH,\boldsymbol{\rho})$, it is always possible to
  construct, from $\qexp{R}$, a (classical) statistical model
  $\exp{E}_{\qexp{R}}=(\Theta,\Delta_{\sH},\boldsymbol{\alpha}_{\boldsymbol{\rho}})$
  that is completely equivalent to $\qexp{R}$, in the sense that, for
  every finite decision set $\outx$, every payoff function
  $\ell:\Theta\times\outx\to\mathbb{R}$, and every $\outx$-POVM
  $\povm{P}$ on $\sH$, there exists a $\outx$-r.d.f. $\phi$ on
  $\Delta_\sH$ such that $\s(\exp{E}_{\qexp{R}},\outx,\ell,\phi)=
  \s_{\textrm{q}}(\qexp{R},\outx,\ell,\povm{P})$. Such a
  correspondence is obtained by first introducing a sample set
  $\Delta_\sH=\{\delta\}$ with $|\Delta_\sH|=\dim\sH$, so that any
  orthonormal basis for $\sH$ can be indexed by $\Delta_\sH$. Then,
  since all density matrices $\rho_\theta$ are pairwise commuting, an
  orthonormal basis
  $\{|\varphi_\delta\>\in\sH\}_{\delta\in\Delta_\sH}$ for $\sH$
  exists, with respect to which all $\rho_\theta$ are simultaneously
  diagonal. Finally, the family of probability distributions
  $\boldsymbol{\alpha}_{\boldsymbol{\rho}}=(p_\theta;\theta\in\Theta)$
  on $\Delta_\sH$ is defined according to the relation
  $p_\theta(\delta):=\<\varphi_\delta|\rho_\theta|\varphi_\delta\>$,
  for all $\delta\in\Delta_\sH$ and $\theta\in\Theta$. Then, it is
  easy to check that the statistical model
  $\exp{E}_{\qexp{R}}=(\Theta,\Delta_\sH,\boldsymbol{\alpha}_{\boldsymbol{\rho}})$,
  obtained in this way from $\qexp{R}=(\Theta,\sH,\boldsymbol{\rho})$,
  is completely equivalent to the initial quantum statistical model
  $\qexp{R}$, in the sense explained above. This in particular implies
  that, for every finite decision set $\outx$ and every payoff function
  $\ell:\Theta\times\outx\to\mathbb{R}$,
  $\$(\exp{E}_{\qexp{R}},\outx,\ell)=
  \$_{\textrm{q}}(\qexp{R},\outx,\ell)$.

  Conversely, given a (classical) statistical model
  $\exp{E}=(\Theta,\Delta,\boldsymbol{\alpha})$, it is always possible
  to construct an equivalent abelian quantum statistical model
  $\qexp{R}_{\exp{E}}=(\Theta,\sH_{\Delta},\boldsymbol{\rho}_{\boldsymbol{\alpha}})$,
  by introducing a Hilbert space $\sH_\Delta$, with
  $\dim\sH_\Delta=|\Delta|$, and a family
  $\boldsymbol{\rho}_{\boldsymbol{\alpha}}=(\rho_\theta;\theta\in\Theta)$
  of diagonal density matrices on $\sH_\Delta$, defined by the
  relation
  $\rho_\theta=\sum_{\delta\in\Delta}p_\theta(\delta)|\varphi_\delta\>\<\varphi_\delta|$,
  where $\{|\varphi_{\delta}\>\in\sH_\Delta\}_{\delta\in\Delta}$ is
  any orthonormal basis for $\sH_\Delta$. Also in this case, it is
  easy to check that the quantum statistical model
  $\qexp{R}_{\exp{E}}=(\Theta,\sH_{\Delta},\boldsymbol{\rho}_{\boldsymbol{\alpha}})$,
  obtained in this way from
  $\exp{E}=(\Theta,\Delta,\boldsymbol{\alpha})$, is completely
  equivalent to the initially given statistical model $\exp{E}$, in
  the sense that, for every finite decision set $\outx$, every payoff
  function $\ell:\Theta\times\outx\to\mathbb{R}$, and every
  $\outx$-r.d.f. $\phi$ on $\Delta$, there exists a $\outx$-POVM
  $\povm{P}$ on $\sH_\Delta$ such that
  $\s_{\textrm{q}}(\qexp{R}_{\exp{E}},\outx,\ell,\povm{P})=
  \s(\exp{E},\outx,\ell,\phi)$. In particular,
  $\$_{\textrm{q}}(\qexp{R}_{\exp{E}},\outx,\ell)=
  \$(\exp{E},\outx,\ell)$.
\end{remark}

These ideas can be compactly re-expressed as follows:

\begin{postulate}[Correspondence principle]\label{post:one}
  Classical statistical models are identified with abelian
  quantum statistical models, and viceversa.
\end{postulate}

A quantum statistical model $\qexp{R}$ involves a parameter set
$\Theta$ and a Hilbert space $\sH$. In a sense, then, a quantum
statistical model constitutes an asymmetric structure, where a quantum
system carries information about a classical parameter. It is useful
hence to provide a notion for a ``fully quantum'' information
structure. This can be done as follows: given a finite parameter set
$\Theta$, let $\sH_\Theta$ be a Hilbert space such that
$\dim\sH_\Theta=|\Theta|$, i.e., such that that there exists a
complete set of orthonormal vectors
$\{|\varphi_\theta\>\in\sH_\Theta\}_{\theta\in \Theta}$, labeled by
$\theta$, which form a basis for $\sH_\Theta$. For the sake of
notation, let us denote $|\varphi_\theta\>$ simply by $|\theta\>$.
Then, each quantum model $\qexp{R}=(\Theta,\sH,\boldsymbol{\rho})$
defines a corresponding bipartite quantum state
\begin{equation}\label{eq:bip-exp}
  \rho_{AB}^{\qexp{R}}:=\frac 1{|\Theta|}\sum_{\theta\in\Theta}|\theta\>\<\theta|_A\otimes\rho^\theta_B,
\end{equation}
where $\sH_A\cong\sH_\Theta$, $\sH_B\cong\sH$, and
$\rho^\theta_B\equiv \rho_\theta$. The particular
``classical-quantum'' structure of the state given
in~(\ref{eq:bip-exp}) reflects the above mentioned ``hybrid''
structure of a quantum statistical model. Instead, by allowing
$\rho_{AB}$ to be an arbitrary bipartite state, we arrive at the
following definition:

\begin{definition}[Quantum Information Structures~\cite{shmaya}]
  A \emph{quantum information structure} $\boldsymbol{\varrhoup}_{AB}$
  is defined as a triple $(\sH_A,\sH_B,\rho_{AB})$, where $\sH_A$ and
  $\sH_B$ are finite dimensional Hilbert spaces, and
  $\rho_{AB}\in\sS(\sH_A\otimes\sH_B)$.
\end{definition}

The notion of quantum information structure is hence the ``fully
quantized'' analogue of a quantum statistical model. In the same way
in which a quantum statistical model can be used to define a quantum
statistical decision problem, a quantum information structure can be
used to define a quantum game\footnote{In the very specific sense
  given in Ref.~\cite{shmaya}.} as follows:

\begin{definition}[Quantum Statistical Decision Games~\cite{shmaya}]
  A \emph{quantum statistical decision game} is defined as a triple
  $(\boldsymbol{\varrhoup}_{AB},\outx,\mathrm{O}_A^\outx)$, where
  $\boldsymbol{\varrhoup}_{AB}=(\sH_A,\sH_B,\rho_{AB})$ is a quantum
  information structure, $\outx$ is a (finite) decision set
  $\{i\}_{i\in\outx}$, and $\mathrm{O}_A^\outx$ is a family
  $(O^i_A;i\in\outx)$ of self-adjoint \emph{payoff operators}
  $O^i_A\in\bound(\sH_A)$. (From now on, the superscript $\outx$ in
  $\mathrm{O}_A^\outx$ will be dropped when clear from the context.)
  The choice of a \emph{strategy} for player $B$ corresponds to the
  choice of a POVM $\povm{P}_B=(P^i_B;i\in\outx)$ on $\sH_B$. The
  corresponding expected payoff is computed as
\begin{equation}
  \s_{\textrm{q}}(\boldsymbol{\varrhoup}_{AB},\outx,\mathrm{O}_A,\povm{P}_B):=\sum_{i\in\outx}\Tr\left[\left(O^i_A\otimes
      P^i_B\right)\rho_{AB}\right].
\end{equation}
The maximum expected payoff is given by
\begin{equation}\label{eq:q-payoff}
  \$_{\textrm{q}}(\boldsymbol{\varrhoup}_{AB},\outx,\mathrm{O}_A):=\max_{\povm{P}_B}
  \sum_{i\in\outx}\Tr\left[\left(O^i_A\otimes
      P^i_B\right)\rho_{AB}\right].
\end{equation}
\end{definition}

The following definition was introduced in~\cite{shmaya} as a very
natural analogue of Definition~\ref{def:bayes_order}:

\begin{definition}[Information Ordering]\label{def:info-ord-qis}
  Given two quantum information structures
  $\boldsymbol{\varrhoup}_{AB}=(\sH_A,\sH_B,\rho_{AB})$ and
  $\boldsymbol{\varsigmaup}_{AB'}=(\sH_A,\sH_{B'},\sigma_{AB'})$, $\boldsymbol{\varrhoup}_{AB}$ is
  said to be \emph{always more informative than} $\boldsymbol{\varsigmaup}_{AB'}$, in
  formula,
\begin{equation}
  \boldsymbol{\varrhoup}_{AB}\supset_A\boldsymbol{\varsigmaup}_{AB'},
\end{equation}
if and only if, for every finite decision set $\outx$ and every family
of self-adjoint payoff operators
$\mathrm{O}_A=(O^i_A;i\in\outx)$ on $\sH_A$,
\begin{equation}
  \$_{\textrm{q}}(\boldsymbol{\varrhoup}_{AB},\outx,\mathrm{O}_A)\ge \$_{\textrm{q}}(\boldsymbol{\varsigmaup}_{AB'},\outx,\mathrm{O}_A).
\end{equation}
\end{definition}

\begin{remark}\label{rem:from-str-to-qmod}
  In analogy with Remark~\ref{rem:quantum-vs-classical}, here we note
  that any quantum information structure
  $\boldsymbol{\varrhoup}_{AB}=(\sH_A,\sH_B,\rho_{AB})$, for which a
  decomposition like that in Eq.~(\ref{eq:bip-exp}) exists, naturally
  induces a corresponding quantum statistical model
  $\qexp{R}_{\boldsymbol{\varrhoup}}=(\Theta,\sH_B,(\rho^\theta_B;\theta\in\Theta))$,
  where the states $\rho^\theta_B$ are those appearing
  in~(\ref{eq:bip-exp}). Moreover, any quantum statistical decision
  game $(\boldsymbol{\varrhoup}_{AB},\outx,\mathrm{O}_A)$ built upon
  such a classical-quantum structure $\boldsymbol{\varrhoup}_{AB}$ is
  completely equivalent to a quantum statistical decision problem
  $(\qexp{R}_{\boldsymbol{\varrhoup}},\outx,\ell_{\mathrm{O}})$, in
  the sense that
  $\$_{\textrm{q}}(\boldsymbol{\varrhoup}_{AB},\outx,\mathrm{O}_A)=
  \$_{\textrm{q}}(\qexp{R}_{\boldsymbol{\varrhoup}},\outx,\ell_{\mathrm{O}})$,
  where the payoff function $\ell_{\mathrm{O}}$ is defined by
  $\ell(\theta,i):=\<\theta_A|O^i_A|\theta_A\>$, with the vectors
  $|\theta_A\>$ being the same as in~(\ref{eq:bip-exp}).
\end{remark}

For the reader's convenience, we end this section by summarizing the
contents of Remarks~\ref{rem:quantum-vs-classical}
and~\ref{rem:from-str-to-qmod} as follows:
\begin{enumerate}
\item the most general notion is that of quantum statistical decision
  games over quantum information structures;
\item quantum statistical decision problems over quantum statistical
  models are equivalent to quantum statistical decision games over
  hybrid classical-quantum information structures;
\item classical statistical decision problems over statistical
  models are equivalent to quantum decision problems over abelian
  quantum statistical models.
\end{enumerate}
In other words, quantum information structures contain quantum
statistical models (as hybrid structures), which, in turn, contain
classical statistical models (as abelian models). For this reason we
will first formulate our results for quantum information structures:
quantum statistical models and classical statistical models will be
considered afterwards, as particular cases.

\section{Sufficiency conditions for quantum information structures and
  statistical morphisms}\label{sec:proc}

In the previous section we extended the notion of information ordering
to quantum statistical models and quantum information structures,
depending on their ``information value'' in quantum statistical
decision problems and quantum statistical decision games,
respectively. In the following we will carefully define what it means
that a quantum information structure is sufficient for another. In
order to do this, we will need to consider a formalism which is
slightly more general than the one we used before.

We begin with the following definition:

\begin{definition}[State Spaces]
  The \emph{state space} $\sS$ of a quantum system defined on a
  Hilbert space $\sH$ is a non-empty subset of $\sS(\sH)$, containing
  all possible physical states of the system.
\end{definition}

\begin{remark} Usually, the state space $\sS$ coincides with the set
  $\sS(\sH)$ of all possible density matrices acting on
  $\sH$. However, there are cases in which the states accessible to
  the system form a proper subset of $\sS(\sH)$, for example, when the
  system is known to obey a conservation law. For later convenience,
  we keep our definition of state space as general as possible. This
  is also the reason why, in the above definition, there is no
  assumption about the convexity of the state space, as we do not need
  such assumption in general (even though, in many physically relevant
  situations, that would seem rather natural).
\end{remark}

\begin{definition}[Effects and Tests]\label{def:effects}
  An operator $X\in\bound(\sH)$ is called an \emph{effect} on a state
  space $\sS$ (defined on $\sH$) if and only if there exists an
  operator $P\in\bound(\sH)$, with $0\le P\le\openone$, such that
  $\Tr[X\rho]=\Tr[P\rho]$, for all $\rho\in\sS$.

  For any (finite) decision set $\outx=\{i\}$, a family
  $(M^i;i\in\outx)$ of operators $M^i\in\bound(\sH)$ is called a
  \emph{$\outx$-test} $\test{M}^\outx$ on a state space $\sS$ (defined
  on $\sH$) if and only if there exists a $\outx$-POVM
  $\povm{P}^\outx=(P^i;i\in\outx)$ on $\sH$ with $\Tr[M^i\
  \rho]=\Tr[P^i\ \rho]$, for all $i\in\outx$ and for all
  $\rho\in\sS$. Any such POVM $\povm{P}^\outx$ is said to realize the
  test $\test{M}^\outx$ on $\sS$. From now on, the superscript $\outx$
  will be dropped when clear from the context.
\end{definition}

\begin{remark}\label{rem:stat-equivalence}
  For a given state space $\sS$ and a given decision set $\outx$, two
  families $\test{M}=(M^i;i\in\outx)$ and $\test{N}=(N^i;i\in\outx)$
  of operators in $\bound(\sH)$ are statistically equivalent on $\sS$,
  in formula $\test{M}\sim_\sS\test{N}$, if and only if $\Tr[M^i\
  \rho]=\Tr[N^i\ \rho]$, for all $i\in\outx$ and all $\rho\in\sS$. For
  any family $\test{M}=(M^i;i\in\outx)$, let $[\test{M}]_\sS$ be the
  corresponding equivalence class induced by $\sim_\sS$. Any
  $\outx$-test on $\sS$ can hence be thought of as the equivalence
  class of some $\outx$-POVM on $\sH$.
\end{remark}

\begin{remark}\label{rem:imp-conter}
  A second, more physically motivated way to think of tests is the
  following: $\outx$-tests on a state space $\sS$ are those affine
  mappings, from $\sS$ to probability distributions on $\outx$, which
  can be \emph{physically realized} as quantum measurements. This is
  guaranteed by requiring, in the definition of test, the existence of
  at least one POVM that is statistically equivalent to it: in fact,
  all physically realizable quantum measurements give rise to a POVM,
  and any POVM can be physically measured~\cite{ozawa-instr}. Such a
  restriction in the definition of tests is meaningful only if there
  exist cases of affine mappings from a state space $\sS$ to
  probability distributions on a decision set $\outx$, which cannot be
  realized by any POVM. If the state space is the totality of states
  $\sS(\sH)$, then, all such affine mappings are in one-to-one
  correspondence with POVM's, and there is no need to introduce
  further definitions. However, in the general case in which
  $\sS\subset\sS(\sH)$, the distinction between tests and
  ``unphysical'' affine mappings become relevant, and
  Definition~\ref{def:effects} is necessary.
\end{remark}

We are now in the position to rigorously introduce the idea which will
be the basis of our analysis:

\begin{definition}[Statistical Morphisms]\label{def:process}
  Given two state spaces $\sS_{\insym}$ (defined on a Hilbert space
  $\sH_{\insym}$) and $\sS_{\outsym}$ (defined on a Hilbert space
  $\sH_{\outsym}$), we say that a linear map
  $\mL:\bound(\sH_{\insym})\to\bound(\sH_{\outsym})$ induces a
  \emph{statistical morphism} $\mL:\sS_{\insym}\to\sS_{\outsym}$ if
  and only if the following conditions are both satisfied:
  \begin{enumerate}
  \item for every $\rho\in\sS_{\insym}$, $\mL(\rho)\in \sS_{\outsym}$;
\item the dual transformation
  $\mL^*:\bound(\sH_{\outsym})\to\bound(\sH_{\insym})$, defined by
  trace duality\footnote{For any operator $X\in\bound(\sH_{\outsym})$,
    $\mL^*(X)\in\bound(\sH_{\insym})$ is defined by the relation
    $\Tr[\mL^*(X)\ Y]=\Tr[X\ \mL(Y)]$, for every
    $Y\in\bound(\sH_{\insym})$.}, maps tests on $\sS_{\outsym}$ into
  tests on $\sS_{\insym}$.
\end{enumerate}
\end{definition}

\begin{remark}\label{rem:pos-morph}
  Notice that the notion of statistical morphism, introduced in
  Definition~\ref{def:process}, is in principle strictly weaker than
  the notion of positive map, which is a linear map that transforms
  positive operators into positive operators. In fact, given a
  positive operator $P\le\openone$ on $\sH_{\outsym}$, the operator
  $\mL^*(P)$ might have negative eigenvalues, and yet be an effect on
  $\sS_{\insym}$, according to Definition~\ref{def:effects}. On the
  contrary, a linear, trace-preserving, positive map from
  $\bound(\sH_{\insym})$ to $\bound(\sH_{\outsym})$ always constitutes
  a statistical morphism. An open question is whether any statistical
  morphism can always be extended to a positive map. Indications that
  this might not be true in general are provided in
  Ref.~\cite{matsumoto}, Corollary~10. However, at the moment of
  writing, an explicit counterexample is not available.
\end{remark}


\begin{definition}\label{def:localss}
  Given a quantum information structure
  $\boldsymbol{\varrhoup}_{AB}=(\sH_A,\sH_B,\rho_{AB})$, the associated
  state space $\sS_B(\boldsymbol{\varrhoup}_{AB})\subseteq\sS(\sH_B)$ of
  physical states of the subsystem $B$ is defined as
  \begin{equation}
    \sS_B(\boldsymbol{\varrhoup}_{AB}):=\left\{\left.\frac{\Tr_A[(P_A\otimes
          \openone_B)\rho_{AB}]}{\Tr[(P_A\otimes
          \openone_B)\rho_{AB}]}\right|P_A\in\bound(\sH_A):0\le P_A\le\openone_A\right\}.
\end{equation}
\end{definition}

\begin{remark}
  The state spaces associated with a given information structures turn
  out to be convex state spaces. This can be easily verified by direct
  inspection.
\end{remark}

\begin{remark}\label{rem:povm-to-tests}
  From now on, it is convenient to think that, in
  Eq.~(\ref{eq:q-payoff}), the maximum over POVM's $\povm{P}_B$ on
  $\sH_B$ is replaced by a maximum over tests $\test{M}_B$ on
  $\sS_B(\boldsymbol{\varrhoup}_{AB})$. Such a replacement, which is
  formally convenient, is quantitatively irrelevant, since it does not
  affect the value of the maximum expected payoff, nor it modifies the
  information ordering relation introduced in
  Definition~\ref{def:info-ord-qis}.
\end{remark}

We are now able to rigorously define the notion of sufficiency (in a
sense analogous to the one used by Blackwell in~\cite{blackwell}) for
quantum information structures, in its two variants: sufficiency and
m-sufficiency.

\begin{definition}[Sufficiency and m-sufficiency]\label{def:suff}
  Given two quantum information structures
  $\boldsymbol{\varrhoup}_{AB}=(\sH_A,\sH_B,\rho_{AB})$ and
  $\boldsymbol{\varsigmaup}_{AB'}=(\sH_A,\sH_{B'},\sigma_{AB'})$, we
  say that $\boldsymbol{\varrhoup}_{AB}$ is \emph{m-sufficient for}
  $\boldsymbol{\varsigmaup}_{AB'}$, in formula
\begin{equation}
  \boldsymbol{\varrhoup}_{AB}\msucc\boldsymbol{\varsigmaup}_{AB'},
\end{equation}
if and only if there exists a statistical morphism
$\mL_B:\sS_B(\boldsymbol{\varrhoup}_{AB})\to\sS_{B'}(\boldsymbol{\varsigmaup}_{AB'})$
such that
\begin{equation}
\sigma_{AB'}=(\id_A\otimes\mL_B)(\rho_{AB}).
\end{equation}
We say that $\boldsymbol{\varrhoup}_{AB}$ is \emph{sufficient for}
$\boldsymbol{\varsigmaup}_{AB'}$, in formula
\begin{equation}
  \boldsymbol{\varrhoup}_{AB}\succ\boldsymbol{\varsigmaup}_{AB'},
\end{equation}
if and only if there exists a completely positive, trace-preserving
map $\mE_B:\bound(\sH_B)\to\bound(\sH_{B'})$ such that
\begin{equation}
\sigma_{AB'}=(\id_A\otimes\mE_B)(\rho_{AB}).
\end{equation}
\end{definition}

Intuitively speaking, the idea of sufficiency is related with the fact
that the transformation can be actually performed \emph{physically},
as an open evolution. On the contrary, the notion of m-sufficiency
introduced here just assumes the existence of a \emph{formal}
statistical procedure to map one strategy into another.

\section{Extension theorems for statistical morphisms}\label{sec:ext}

Even if the notion of statistical morphism is weaker than that of
positive map, two famous extension theorems for positive maps, proved
by Choi~\cite{choi} and Arveson~\cite{arveson}, can be generalized to
statistical morphisms as well.

\begin{definition}[Complete State Spaces]\label{def:complete}
  A state space $\sS$ on $\sH$ is called \emph{complete} for
  $\bound(\sH)$ if and only if it contains $(\dim\sH)^2$ linearly
  independent density matrices.
\end{definition}

\begin{definition}[Composition of State Spaces]\label{def:compos}
  Given two state spaces $\sS_{\alpha}$ (on $\sH_{\alpha}$) and
  $\sS_{\beta}$ (on $\sH_{\beta}$), we define the set
  \begin{equation}
    \sS_{\alpha}\times\sS_{\beta}:=\left\{\sigma_{\alpha}\otimes\tau_{\beta}\left|\sigma_{\alpha}\in\sS_{\alpha},\tau_{\beta}\in\sS_{\beta}\right.\right\}.
  \end{equation}
  An operator $X\in\bound(\sH_\alpha\otimes\sH_\beta)$ is an effect on
  $\sS_{\alpha}\times\sS_{\beta}$ if and only if there exists an
  operator $P\in\bound(\sH_\alpha\otimes\sH_\beta)$, $0\le
  P\le\openone_{\alpha}\otimes\openone_{\beta}$, such that $\Tr[X\
  (\sigma_{\alpha}\otimes\tau_{\beta})] =\Tr[P\
  (\sigma_{\alpha}\otimes\tau_{\beta})]$, for all
  $\sigma_{\alpha}\in\sS_{\alpha}$ and
  $\tau_{\beta}\in\sS_{\beta}$. In the same way we extend the notion
  of tests. Notice that effects or tests on
  $\sS_{\alpha}\times\sS_{\beta}$ need not be factorized.
\end{definition}

\begin{proposition}\label{prop:ent_ext}
  Given two state spaces $\sS_{\insym}$ and $\sS_{\outsym}$, defined
  on $\sH_{\insym}$ and $\sH_{\outsym}$, respectively, and a third
  auxiliary complete state space $\sS_0$, defined on $\sH_0\cong
  \sH_{\outsym}$, suppose that the linear map
  $\id\otimes\mL:\bound(\sH_0)\otimes\bound(\sH_{\insym})
  \to\bound(\sH_0)\otimes\bound(\sH_{\outsym})$ induces a statistical
  morphism from $\sS_0\times\sS_{\insym}$ to
  $\sS_0\times\sS_{\outsym}$. Then, there exists a completely
  positive, trace-preserving map $\mE:\bound(\sH_{\insym})\to
  \bound(\sH_{\outsym})$ such that
  \begin{equation}
    \mL(\sigma)=\mE(\sigma),
  \end{equation}
  for all $\sigma\in\sS_{\insym}$.
\end{proposition}

\begin{proof}
  Let $(B^i)_{i=1}^{d^2}$, where $d=\dim\sH_0=\dim\sH_{\outsym}$, be
  the POVM consisting of the $d^2$ generalized Bell projectors acting
  on $\sH_0\otimes\sH_{\outsym}$. By trace-duality:
  \begin{equation}
    \Tr\left[B^i(\omega\otimes\mL(\sigma))\right]=\Tr\left[(\id\otimes\mL^*)(B^i)\
      (\omega\otimes\sigma)\right],
  \end{equation}
  for all $\sigma\in\sS_{\insym}$ and all $\omega\in\sS_0$. The fact
  that $\id\otimes\mL$ is a statistical morphism implies, by
  definition, that the operators
  $((\id\otimes\mL^*)(B^i))_{i=1}^{d^2}$, even if not positive, yet
  induce a test on $\sS_0\times\sS_{\insym}$. In other words, there
  exists a POVM $(\tilde B^i)_{i=1}^{d^2}$ on
  $\sH_0\otimes\sH_{\insym}$ such that
    \begin{equation}\label{eq:asd}
      \Tr\left[(\id\otimes\mL^*) (B^i)\
        (\omega\otimes\sigma)\right]=\Tr\left[\tilde B^i\ (\omega\otimes\sigma)\right],
    \end{equation}
    for all $\sigma\in\sS_{\insym}$, all $\omega\in\sS_0$, and
    every $i$. Due to the assumption that $\sS_0$ is complete, there
    always exist $d^2$ states in $\sS_0$ which form an operator
    basis for $\bound(\sH_0)$. We can then extend Eq.~(\ref{eq:asd}) by
    linearity and obtain that, in fact,
\begin{equation}
  \Tr\left[B^i\ (X\otimes\mL(\sigma))\right]=\Tr\left[\tilde B^i\ (X\otimes\sigma)\right],
    \end{equation}
    for all $\sigma\in\sS_{\insym}$, all $X\in\bound(\sH_0)$, and
    every $i$.

    Using the POVM $(\tilde B^i)_{i=1}^{d^2}$ (whose existence we
    proved above), we now consider the identity (via teleportation):
\begin{equation}\label{eq:telep}
\begin{split}
&\mL(\sigma)\\
=&\sum_{i=1}^{d^2}\Tr_{\beta\gamma}\left[\left(U^i_{\alpha}\otimes\openone_{\beta\gamma}\right)\left(\openone_{\alpha}\otimes
    B^i_{\beta\gamma}\right)\left(\Psi^+_{\alpha\beta}\otimes\mL_{\gamma}(\sigma_{\gamma})\right)\left(
  (U^i_{\alpha})^{\dag}\otimes\openone_{\beta\gamma}\right)\right]\\
=&\sum_{i=1}^{d^2}\Tr_{\beta\gamma}\left[\left(U^i_{\alpha}\otimes\openone_{\beta\gamma}\right)\left(\openone_{\alpha}\otimes
  \tilde B^i_{\beta\gamma}\right)\left(\Psi^+_{\alpha\beta}\otimes\sigma_{\gamma}\right)\left(
  (U^i_{\alpha})^{\dag}\otimes\openone_{\beta\gamma}\right)\right],
\end{split}
\end{equation}
where $\Psi^+=d^{-1}\sum_{i,j=1}^d|i\>\<j|\otimes|i\>\<j|$ is a
maximally entangled state on $\sH_0^{\otimes 2}$ and
$(U^i)_{i=1}^{d^2}$ is an appropriate family of unitary matrices on
$\sH_0$.  The relation above holds for all
$\sigma\in\sS_{\insym}\subseteq\sS(\sH_{\insym})$. However, it is clear
that the last term in Eq.~(\ref{eq:telep}) can be extended, by
linearity, to a completely positive trace-preserving map
$\mE:\bound(\sH_{\insym})\to \bound(\sH_0)\cong\bound(\sH_{\outsym})$
defined as:
\begin{equation}
\begin{split}
  &\mE(\rho)\\
  :=&\sum_{i=1}^{d^2}\Tr_{\beta\gamma}\left[\left(U^i_{\alpha}\otimes\openone_{\beta\gamma}\right)\left(\openone_{\alpha}\otimes
      \tilde B^i_{\beta\gamma}\right)\left(\Psi^+_{\alpha\beta}\otimes\rho_\gamma\right)\left(
      (U^i_{\alpha})^{\dag}\otimes\openone_{\beta\gamma}\right)\right],
\end{split}
\end{equation}
for all $\rho\in\sS(\sH_{\insym})$. This hence concludes the proof
that a completely positive trace-preserving map
$\mE:\bound(\sH_{\insym})\to \bound(\sH_{\outsym})$ exists, such that
\begin{equation}
  \mE(\sigma)=\mL(\sigma),
\end{equation}
for all $\sigma\in\sS_{\insym}$.
\end{proof}

Another important case is when the output state space $\sS_{\outsym}$
is abelian, namely, $[\rho,\sigma]=0$, for all
$\rho,\sigma\in\sS_{\outsym}$. This condition, in particular, implies
that there exists an orthonormal basis $\{|i\>\}_{i=1}^d$ for
$\sH_{\outsym}$ that diagonalizes all $\rho\in\sS_{\outsym}$.

\begin{proposition}\label{prop:class_ext}
  Given two state spaces $\sS_{\insym}$ and $\sS_{\outsym}$, defined
  on $\sH_{\insym}$ and $\sH_{\outsym}$, respectively, let
  $\sS_{\outsym}$ be abelian. If there exists a linear map
  $\mL:\bound(\sH_{\insym})\to\bound(\sH_{\outsym})$ inducing a
  statistical morphism from $\sS_{\insym}$ to $\sS_{\outsym}$, then
  there exists a completely positive, trace-preserving map
  $\mE:\bound(\sH_{\insym})\to \bound(\sH_{\outsym})$ such that
  \begin{equation}
    \mL(\rho)=\mE(\rho),
  \end{equation}
  for all $\rho\in\sS_{\insym}$.
\end{proposition}

\begin{proof}
  For $d=\dim\sH_{\outsym}$, let $\{|i\>\}_{i=1}^{d}$ be the basis for
  $\sH_{\outsym}$ that simultaneously diagonalizes every
  $\sigma\in\sS_{\outsym}$, and denote by
  $\Pi_i\in\bound(\sH_{\outsym})$ each projector $|i\>\<i|$. Then, for
  any $\sigma\in\sS_{\outsym}$
\begin{equation}\label{eq:ident-abe}
\sigma=\sum_{i=1}^d\Tr[\Pi_i\ \sigma]\Pi_i.
\end{equation}
Next, we note that, by defition of the trace-dual map $\mL^*$,
  \begin{equation}
    \Tr\left[\Pi^i\ \mL(\rho)\right]=\Tr\left[\mL^*(\Pi^i)\
      \rho\right],
  \end{equation}
  for all $\rho\in\sS_{\insym}$. The fact that $\mL$ is a statistical
  morphism implies, by definition, that the operators
  $(\mL^*(\Pi^i))_{i=1}^{d}$, even if not positive, yet induce a test
  on $\sS_{\insym}$. In other words, there exists a POVM $(\tilde
  \Pi^i)_{i=1}^{d}$ such that
    \begin{equation}
      \Tr\left[\mL^* (\Pi^i)\ \rho\right]=\Tr\left[\tilde \Pi^i\ \rho\right],
    \end{equation}
    for all $\rho\in\sS_{\insym}$ and every $i$.

    Using the POVM $(\tilde \Pi^i)_{i=1}^{d}$ (whose existence we
    proved above), we recall Eq.~(\ref{eq:ident-abe}) above and
    consider the identity:
\begin{equation}\label{eq:telep2}
\begin{split}
\mL(\rho)&=\sum_{i=1}^{d}\Tr\left[\Pi^i\ \mL(\rho)\right]\Pi^i\\
&=\sum_{i=1}^{d}\Tr\left[\tilde \Pi^i\ \rho\right]\Pi^i,
\end{split}
\end{equation}
The relation above holds for all
$\rho\in\sS_{\insym}\subseteq\sS(\sH_{\insym})$. However, it is clear
that the last term in Eq.~(\ref{eq:telep2}) can be extended, by
linearity, to a completely positive trace-preserving map $\mE:
\bound(\sH_{\insym})\to \bound(\sH_{\outsym})$ defined as:
\begin{equation}
\mE(\rho):=\sum_{i=1}^{d}\Tr\left[\tilde \Pi^i\ \rho\right]\Pi^i,
\end{equation}
for all $\rho\in\sS(\sH_{\insym})$. This hence concludes the proof that
a completely positive trace-preserving map
$\mE:\bound(\sH_{\insym})\to \bound(\sH_{\outsym})$ exists, such that
\begin{equation}
  \mE(\rho)=\mL(\rho),
\end{equation}
for all $\rho\in\sS_{\insym}$.
\end{proof}

\section{A fundamental equivalence relation}\label{sec:main}

In this section, we prove our main result:

\begin{theorem}\label{thm:main}
  Given two quantum information structures
  $\boldsymbol{\varrhoup}_{AB}=(\sH_A,\sH_B,\rho_{AB})$ and
  $\boldsymbol{\varsigmaup}_{AB'}=(\sH_A,\sH_{B'},\sigma_{AB'})$,
  \begin{equation}
\boldsymbol{\varrhoup}_{AB}\msucc\boldsymbol{\varsigmaup}_{AB'} \ \Leftrightarrow\    \boldsymbol{\varrhoup}_{AB}\supset_A\boldsymbol{\varsigmaup}_{AB'}.
\end{equation}
Moreover, the linear map inducing the statistical morphism between
$\boldsymbol{\varrhoup}_{AB}$ and $\boldsymbol{\varsigmaup}_{AB'}$ can
always be chosen to be trace-preserving on the whole space
$\bound(\sH_B)$.
\end{theorem}

\begin{remark}
  Shmaya, in Remark~7 of his Ref.~\cite{shmaya}, asks the question
  whether
  $\boldsymbol{\varrhoup}_{AB}\supset_A\boldsymbol{\varsigmaup}_{AB'}$
  is equivalent to the existence of a positive trace-preserving map
  $\mathcal{P}$ such that
  $\sigma_{AB'}=(\id\otimes\mathcal{P})\rho_{AB}$. The above theorem
  shows that Shmaya's question is equivalent to asking whether any
  trace-preserving statistical morphism always admits a
  trace-preserving positive extension (about this point, see
  Remark~\ref{rem:pos-morph} above).
\end{remark}

For the sake of clarity, we divide the proof of Theorem~\ref{thm:main}
in two parts. The first part is a lemma proved by Shmaya in
Ref.~\cite{shmaya}, as a direct consequence of the Separation Theorem
for convex sets (see, e.~g., Ref.~\cite{rocka}).

Before stating the lemma, we introduce the following notation: given a
quantum information structure
$\boldsymbol{\varrhoup}_{AB}=(\sH_A,\sH_B,\rho_{AB})$, a decision set
$\outx$, and a test $\test{M}_B=(M^i_B;i\in\outx)$ on the state
space $\sS_B(\boldsymbol{\varrhoup}_{AB})$, we define the following
operators:
\begin{equation}\label{eq:reduction}
  \rho^i_{A|\test{M}}:=\Tr_B\left[(\openone_A\otimes M^i_B)\ \rho_{AB}\right],
\end{equation}
for each $i\in\outx$. In Eq.~(\ref{eq:reduction}), we can replace the
family of operators $\test{M}_B$ by any other family of operators
which is statistically equivalent (in the sense of
Remark~\ref{rem:stat-equivalence}) to $\test{M}_B$ on
$\sS_B(\boldsymbol{\varrhoup}_{AB})$\footnote{This fact can be proved
  by noticing that the joint probability distribution
  $p_{\outy,\outx}(j,i):=\Tr[(F^j_A\otimes M^i_B)\ \rho_{AB}]$, where
  $(F^j_A;j\in\outy)$ is an informationally complete POVM on $\sH_A$,
  equals, for all $j\in\outy$ and all $i\in\outx$, that computed as
  $\Tr[(F^j_A\otimes X^i_B)\ \rho_{AB}]$, whenever
  $(X^i_B;i\in\outx)\sim_{\sS_B(\boldsymbol{\varrhoup}_{AB})}
  (M^i_B;i\in\outx)$. By the completeness of $(F^j_A;j\in\outy)$, we
  conclude that, in fact, $\Tr_B[(\openone_A\otimes M^i_B)\
  \rho_{AB}]=\Tr_B[(\openone_A\otimes X^i_B)\ \rho_{AB}]$, for all
  $i\in\outx$.}. In particular, we can replace the operators $M^i_B$
by the elements $P^i_B$ of any POVM $\povm{P}_B=(P^i_B;i\in\outx)$ on
$\sH_B$ realizing the test $\test{M}_B$ on
$\sS_B(\boldsymbol{\varrhoup}_{AB})$.

We are now ready to state the following:
\begin{lemma}[Shmaya~\cite{shmaya}]\label{lemma:shmaya}
  Given two quantum information structures
  $\boldsymbol{\varrhoup}_{AB}=(\sH_A,\sH_B,\rho_{AB})$ and
  $\boldsymbol{\varsigmaup}_{AB'}=(\sH_A,\sH_{B'},\sigma_{AB'})$, if
  $\boldsymbol{\varrhoup}_{AB}\supset_A\boldsymbol{\varsigmaup}_{AB'}$, then, for any finite decision
  set $\outx$ and any test
  $\test{N}_{B'}=(N^i_{B'};i\in\outx)$ on
  $\sS_{B'}(\boldsymbol{\varsigmaup}_{AB'})$, there exists a test
  $\overline{\test{M}}_B=\left(\overline{M}^i_B;i\in\outx\right)$ on
  $\sS_B(\boldsymbol{\varrhoup}_{AB})$ such that
  \begin{equation}
    \rho_{A\left|\overline{\test{M}}\right.}^i=\sigma_{A|\test{N}}^i
  \end{equation}
  for all $i\in\outx$.
\end{lemma}

\begin{proof}
  For the reader's convenience, we reformulate here Shmaya's proof
  according to our notation. For any finite decision set $\outx$, let us
  consider the set $\mathcal{C}_A(\boldsymbol{\varrhoup}_{AB},\outx)$ of all
  $|\outx|$-tuples
\begin{equation}
  \left(\rho_{A|\test{M}}^1\;,\,\rho_{A|\test{M}}^2\;,\,\cdots\;,\,\rho_{A|\test{M}}^{|\outx|}\right),
\end{equation}
where $\test{M}_B$ varies over all possible $\outx$-tests on
$\sS_B(\boldsymbol{\varrhoup}_{AB})$. Clearly,
$\mathcal{C}_A(\boldsymbol{\varrhoup}_{AB},\outx)$ is a closed and
bounded convex subset of the (real) linear space of $|\outx|$-tuples
$(T^i;i\in\outx)$ of self-adjoint matrices on $\sH_A$, since it
inherits its structure from the convex structure of the set of
$\outx$-tests on $\sS_B(\boldsymbol{\varrhoup}_{AB})$.

The proof then proceeds by \emph{reductio ad absurdum}. Suppose in
fact that, for some decision set $\outx$, there exists a test
$\test{N}_{B'}=(N^i_{B'}\;;\,i\in\outx)$ on $\sS_{B'}(\boldsymbol{\varsigmaup}_{AB'})$
such that the corresponding $|\outx|$-tuple
\begin{equation}
  \left(\sigma_{A|\test{N}}^1\;,\,\sigma_{A|\test{N}}^2\;,\,\cdots\;,\,\sigma_{A|\test{N}}^{|\outx|}\right)\notin\mathcal{C}_A(\boldsymbol{\varrhoup}_{AB}\;,\,\outx).
\end{equation}
Then, by the so-called Separation Theorem between convex sets (see,
e.~g., Ref.~\cite{rocka}, Corollary 11.4.2), there exists a
$|\outx|$-tuple of self-adjoint operators $(\tilde T^i_A;i\in\outx)$
on $\sH_A$, such that
\begin{equation}
  \max_{\test{M}_B}\sum_{i\in\outx}\Tr\left[\rho_{A|\test{M}}^i\
    \tilde T^i\right]<\sum_{i\in\outx}\Tr\left[\sigma_{A|\test{N}}^i\
    \tilde T^i\right],
\end{equation}
where the maximization if taken over all tests $\test{M}_B
=(M^i_B\;;\,i\in\outx)$ on $\sS_B(\boldsymbol{\varrhoup}_{AB})$. This
contradicts the assumption
$\boldsymbol{\varrhoup}_{AB}\supset_A\boldsymbol{\varsigmaup}_{AB'}$.
\end{proof}

\begin{proof}[Proof of Theorem~\ref{thm:main}]
  One direction of the theorem, that is
  $\boldsymbol{\varrhoup}_{AB}\msucc\boldsymbol{\varsigmaup}_{AB'}\
  \Rightarrow\ \boldsymbol{\varrhoup}_{AB}
  \supset_A\boldsymbol{\varsigmaup}_{AB'}$, simply follows from the
  definition of m-sufficiency given in Definition~\ref{def:suff}.

  Only the converse direction, i.e. $\boldsymbol{\varrhoup}_{AB}
  \supset_A\boldsymbol{\varsigmaup}_{AB'}\ \Rightarrow\
  \boldsymbol{\varrhoup}_{AB}\msucc\boldsymbol{\varsigmaup}_{AB'}$, is
  hence non trivial. In order to construct a statistical morphism
  $\mL_B:\sS_B(\boldsymbol{\varrhoup}_{AB})\to
  \sS_{B'}(\boldsymbol{\varsigmaup}_{AB'})$, consider the decision set
  $\outx=\{1,2,\cdots,(\dim\sH_{B'})^2\}$ and an informationally
  complete $\outx$-POVM $(F^i_{B'};i\in\outx)$ on $\sH_{B'}$, with
  self-adjoint dual operators $(\theta^i_{B'};i\in\outx)$. The
  following identity holds
  \begin{equation}
    T_{B'}=\sum _{i\in\outx}\Tr[T_{B'}\ F^i_{B'}]\theta^i_{B'},
  \end{equation}
  for all operators $T_{B'}\in\bound(\sH_{B'})$. By linearity then
  \begin{equation}\label{eq:identity-ex}
    T_{AB'}=\sum _{i\in\outx}\Tr_{B'}\left[T_{AB'}\ (\openone_A\otimes F^i_{B'})\right]\otimes\theta^i_{B'},
  \end{equation}
for all operators $T_{AB'}\in\bound(\sH_A\otimes\sH_{B'})$.

Let us now put, in Eq.~(\ref{eq:identity-ex}),
$T_{AB'}=\sigma_{AB'}$. Since we assume $\boldsymbol{\varrhoup}_{AB}
\supset_A\boldsymbol{\varsigmaup}_{AB'}$, by Lemma~\ref{lemma:shmaya}, there exists a
$\outx$-POVM $(\tilde F^i_B;i\in\outx)$ on $\sH_B$ such that
  \begin{equation}\label{eq:tuple-map}
\Tr_B\left[\rho_{AB}\ (\openone_A\otimes \tilde F^i_B)\right]=   \Tr_{B'}\left[\sigma_{AB'}\ (\openone_A\otimes F^i_{B'})\right],
  \end{equation}
  for all $i\in\outx$. Fixed any such POVM $(\tilde F^i_B;i\in\outx)$,
  we define a linear map $\mL_B:\bound(\sH_B)\to\bound(\sH_{B'})$ via
  the relation
  \begin{equation}\label{eq:mappa}
    \mL_B(T_B):=\sum _{i\in\outx}\Tr[T_B\ \tilde F^i_B]\theta^i_{B'},
  \end{equation}
  for all operators $T_B\in\bound(\sH_B)$. Equivalently, the linear
  map $\mL_B$ can be defined by the relations
  $\mL_{B'}^*(F^i_{B'})=\tilde F^i_B$, for all $i\in\outx$. This
  guarantees $\mL_{B'}^*(\openone_{B'})=\openone_B$, i.e. the linear
  map $\mL_B$ is trace-preserving. Now, we have to check that the
  linear map $\mL_B:\bound(\sH_B)\to\bound(\sH_{B'})$ so constructed
  in fact satisfies both conditions in Definition~\ref{def:process}
  and induces a statistical morphism from
  $\sS_B(\boldsymbol{\varrhoup}_{AB})$ to
  $\sS_{B'}(\boldsymbol{\varsigmaup}_{AB'})$.


  We begin by noting that, as a consequence of
  Eqs.~(\ref{eq:identity-ex}), (\ref{eq:tuple-map}),
  and~(\ref{eq:mappa}),
  $(\id_A\otimes\mL_B)(\rho_{AB})=\sigma_{AB'}$. This can be shown as
  follows:
  \begin{equation}
\begin{split}
  \sigma_{AB'}=&\sum _{i\in\outx}\Tr_{B'}\left[\sigma_{AB'}\
    (\openone_A\otimes F^i_{B'})\right]\otimes\theta^i_{B'}\\
=&\sum _{i\in\outx}\Tr_B\left[\rho_{AB}\ (\openone_A\otimes \tilde
  F^i_B)\right]\otimes\theta^i_{B'}\\
\overset{\textrm{def}}{=}&(\id_A\otimes\mL_B)(\rho_{AB}).
\end{split}  
\end{equation}
This also ensures that
$\mL_B(\sS_B(\boldsymbol{\varrhoup}_{AB}))\subseteq
\sS_{B'}(\boldsymbol{\varsigmaup}_{AB'})$.

Let now $\outx$ be an arbitrary (finite) decision set, and let
$\test{N}_{B'}:=(N^i_{B'};i\in\outx)$ be any $\outx$-test on
$\sS_{B'}(\boldsymbol{\varsigmaup}_{AB'})$. We will now check, by
applying Lemma~\ref{lemma:shmaya}, that the operators
$X^i_B:=\mL^*(N^i_{B'})$ indeed constitute a test on
$\sS_B(\boldsymbol{\varrhoup}_{AB})$. The proof goes as follows: for
every $\omega_B\in\sS_B(\boldsymbol{\varrhoup}_{AB})$, let
$R_A^{\omega}\in\bound(\sH_A)$ be the positive operator such that
$\omega_B=\Tr_A\left[(R_A^{\omega}\otimes\openone_B)\
  \rho_{AB}\right]$. Consider now, for all $i\in\outx$, the trace
\begin{equation}\label{eq:jkl}
\begin{split}
  \Tr[X^i_B\ \omega_B]&=\Tr\left[(R_A^{\omega}\otimes X^i_B)\rho_{AB}\right]\\
  &=\Tr\left[R_A^{\omega}\ \Tr_B\left[(\openone_A\otimes X^i_B)\
      \rho_{AB}\right]\right]\\
  &=\Tr\left[R_A^{\omega}\ \Tr_B\left[(\openone_A\otimes \mL_B^*(N^i_{B'}))\
      \rho_{AB}\right]\right]\\
&=\Tr\left[R_A^{\omega}\ \Tr_{B'}\left[(\openone_A\otimes N^i_{B'})\
      (\id_A\otimes\mL_B)(\rho_{AB})\right]\right]\\
  &=\Tr\left[R_A^{\omega}\ \Tr_{B'}\left[(\openone_A\otimes N^i_{B'})\
      \sigma_{AB'}\right]\right].
\end{split}
\end{equation}
Lemma~\ref{lemma:shmaya} provides the existence of a POVM
$(\overline{P}^i_B;i\in\outx)$ on $\sH_B$ such that
\begin{equation}
  \Tr_B\left[(\openone_A\otimes \overline{P}^i_B)\
    \rho_{AB}\right]=\Tr_{B'}\left[(\openone_A\otimes N^i_{B'})\
    \sigma_{AB'}\right],
\end{equation}
for all $i\in\outx$. Plugging such POVM into Eq.~(\ref{eq:jkl}), we
obtain
\begin{equation}
  \begin{split}
    \Tr[X^i_B\ \omega_B]&=\Tr\left[R_A^{\omega}\
      \Tr_{B'}\left[(\openone_A\otimes N^i_{B'})\
        \sigma_{AB'}\right]\right]\\
    &=\Tr\left[R_A^{\omega}\ \Tr_B\left[(\openone_A\otimes
        \overline{P}^i_B)\ \rho_{AB}\right]\right]\\
&=\Tr\left[(R_A^{\omega}\otimes
        \overline{P}^i_B)\ \rho_{AB}\right]\\
    &=\Tr\left[\overline{P}^i_B\ \omega_B\right],
\end{split}
\end{equation}
for all $i\in\outx$. Since this holds for every
$\omega_B\in\sS_B(\boldsymbol{\varrhoup}_{AB})$, we proved that, for
any finite $\outx$ and any $\outx$-test $(N^i_{B'};i\in\outx)$ on
$\sS_{B'}(\boldsymbol{\varsigmaup}_{AB'})$, the operators
$X^i_B:=\mL_B^*\left(N^i_{B'}\right)$ indeed constitute a test on
$\sS_B(\boldsymbol{\varrhoup}_{AB})$. This shows that $\mL_B$ is a
well-defined statistical morphism from $\sS_B(\boldsymbol{\varrhoup}_{AB})$ to
$\sS_{B'}(\boldsymbol{\varsigmaup}_{AB'})$, as requested.
\end{proof}



\section{The Blackwell-Sherman-Stein theorem in the quantum
  case}\label{sec:black}

The BSS theorem (see Theorem~\ref{theo:blackwell-orig}) is about the
comparison of classical statistical models. According to
Postulate~\ref{post:one}, however, we can actually identify the notion
of classical statistical models with that of abelian quantum
statistical models, so that the BSS theorem becomes a statement about
comparison of abelian quantum statistical models. In this sense then,
we call a ``non-commutative (or quantum) BSS theorem'' a statement
characterizing equivalent conditions for the comparison of general
quantum statistical models, recovering
Theorem~\ref{theo:blackwell-orig} in the abelian case. In the
following, we will show how Theorem~\ref{thm:main} can be used to
prove such a generalized statement.

\begin{definition}
  Given a quantum statistical model
  $\qexp{R}=(\Theta,\sH,\boldsymbol{\rho})$, the associated state
  space $\sS(\qexp{R})\subset\sS(\sH)$ is defined as the set of states
  $\sS(\qexp{R})=\{\rho_\theta:\theta\in\Theta\}$.
\end{definition}

\begin{remark}
  As already noticed in Remark~\ref{rem:povm-to-tests}, it is
  irrelevant whether the maximum in Eq.~(\ref{eq:q-exp-payoff}) is
  taken over POVM's on $\sH$ or over tests on $\sS(\qexp{R})$. For
  what follows, however, it is convenient to consider the expected
  payoff as maximized over tests, rather than POVM's.
\end{remark}

As it happens for quantum information structures (see
Definition~\ref{def:suff}), also for quantum statistical models we
have two different notions of sufficiency:

\begin{definition}[Sufficiency and m-sufficiency]\label{def:qexp_suff}
  The quantum statistical model $\qexp{R}=(\Theta,\sH,\boldsymbol{\rho})$
  is said to be \emph{m-sufficient for}
  $\qexp{S}=(\Theta,\sH',\boldsymbol{\sigma})$, in formula
\begin{equation}
  \qexp{R}\msucc\qexp{S},
\end{equation}
if and only if there exists a statistical morphism
$\mL:\sS(\qexp{R})\to\sS(\qexp{S})$ such that
\begin{equation}
  \sigma_\theta=\mL(\rho_\theta),\qquad\forall\theta\in\Theta.
\end{equation}
The quantum statistical model $\qexp{R}=(\Theta,\sH,\boldsymbol{\rho})$
is said to be \emph{sufficient for}
$\qexp{S}=(\Theta,\sH',\boldsymbol{\sigma})$, in formula
\begin{equation}
  \qexp{R}\succ\qexp{S},
\end{equation}
if and only if there exists a completely positive, trace-preserving
map $\mE:\bound(\sH)\to\bound(\sH')$ such that
\begin{equation}
  \sigma_\theta=\mE(\rho_\theta),\qquad\forall\theta\in\Theta.
\end{equation}
\end{definition}

Theorem~\ref{thm:main}, via the correspondence exhibited in
Eq.~(\ref{eq:bip-exp}), directly implies the following:

\begin{theorem}[Non-commutative BSS Theorem]\label{coro:weak-qexp}
  Given two quantum statistical models
  $\qexp{R}=(\Theta,\sH,\boldsymbol{\rho})$ and
  $\qexp{S}=(\Theta,\sH',\boldsymbol{\sigma})$,
  \begin{equation}
\qexp{R}\msucc\qexp{S}    \ \Leftrightarrow\ \qexp{R}\supset\qexp{S}.
\end{equation}
\end{theorem}

\begin{proof}
  Given the quantum statistical model
  $\qexp{R}=(\Theta,\sH,\boldsymbol{\rho})$, let us construct the quantum
  information structure
  $\boldsymbol{\varrhoup}_{AB}^{\qexp{R}}=(\sH_A,\sH_B,\rho_{AB})$, as
  done in Eq.~(\ref{eq:bip-exp}). Let us repeat the same construction
  (using the same basis for $\sH_A\cong\sH_\Theta$) to obtain
  $\boldsymbol{\varsigmaup}_{AB'}^{\qexp{S}}=(\sH_A,\sH_{B'},\sigma_{AB})$
  from $\qexp{S}=(\Theta,\sH',\boldsymbol{\sigma})$. Keeping in mind
  Remark~\ref{rem:from-str-to-qmod}, it is easy to verify that
\begin{equation}
   \qexp{R}\msucc\qexp{S}\ \Leftrightarrow\ \boldsymbol{\varrhoup}_{AB}^{\qexp{R}}\msucc\boldsymbol{\varsigmaup}_{AB'}^{\qexp{S}},
\end{equation}
and that
\begin{equation}
  \qexp{R}\supset\qexp{S}\ \Leftrightarrow\ \boldsymbol{\varrhoup}_{AB}^{\qexp{R}}\supset_A\boldsymbol{\varsigmaup}_{AB'}^{\qexp{S}}.
\end{equation}
We then obtain the statement by direct application of Theorem~\ref{thm:main}.
\end{proof}

Further, by applying Proposition~\ref{prop:class_ext}, we obtain the
following:

\begin{proposition}[Semi-classical case]\label{coro:blackwell}
  Given two quantum statistical models
  $\qexp{R}=(\Theta,\sH,\boldsymbol{\rho})$ and
  $\qexp{S}=(\Theta,\sH',\boldsymbol{\sigma})$, if $\qexp{S}$ is abelian,
  \begin{equation}
  \qexp{R}\succ\qexp{S} \ \Leftrightarrow\  \qexp{R}\supset\qexp{S}.
\end{equation}
\end{proposition}

\begin{proof}
  By definition, $\qexp{S}$ is an abelian quantum statistical model if
  and only if $\sS(\qexp{S})$ is an abelian state space. Then, due to
  Proposition~\ref{prop:class_ext}, we know that, whenever $\qexp{S}$
  is an abelian quantum statistical model, $\qexp{R}\msucc\qexp{S}$ if
  and only if $\qexp{R}\succ\qexp{S}$. With these remarks at hand, the
  statement is finally proved as a simple consequence of
  Theorem~\ref{coro:weak-qexp} above.
\end{proof}

Notice that Proposition~\ref{coro:blackwell} is still more general
than the BSS theorem, since commutativity is required only for
$\qexp{S}$, whereas the classical case is equivalent to the situation
in which both $\qexp{R}$ and $\qexp{S}$ are
abelian. Proposition~\ref{coro:blackwell} hence describes a
``semi-classical'' scenario. In the case in which also $\qexp{R}$ is
an abelian quantum statistical model, it is easy to prove that any
completely positive, trace-preserving map $\mE$ such that
$\sigma_\theta=\mE(\rho_\theta)$ can be in fact written as a
transition matrix $\mathrm{M}_\mE$, mapping the vectors $\vec
p_\theta$ of eigenvalues of $\rho_\theta$ into the vectors $\vec
q_\theta$ of eigenvalues of $\sigma_\theta$, for all
$\theta\in\Theta$, in complete accordance with the notion of
sufficiency used in the BSS theorem~\ref{theo:blackwell-orig}. We
leave the proof of this to the reader.

Next, we show that Theorem~\ref{coro:weak-qexp}, together with
Proposition~\ref{prop:ent_ext}, provides an equivalent
characterization of the sufficiency relation $\succ$ for quantum
statistical models. We first need the following definitions:

\begin{definition}[Composition of Quantum Statistical Models]
  Given any two quantum statistical models
  $\qexp{R}=(\Theta,\sH,\boldsymbol{\rho})$, with
  $\boldsymbol{\rho}=(\rho_\theta;\theta\in\Theta)$, and
  $\qexp{T}=(\Xi,\mathcal{K},\boldsymbol{\tau})$, with
  $\boldsymbol{\tau}=(\tau_\xi;\xi\in\Xi)$, the composition
  $\qexp{T}\times\qexp{R}$ is defined as the quantum statistical model
  $(\Xi\times\Theta,\mathcal{K}\otimes\sH,\boldsymbol{\tau}\times
  \boldsymbol{\rho})$, where $\boldsymbol{\tau}\times
  \boldsymbol{\rho}:=(\tau_\xi\otimes
  \rho_\theta;\xi\in\Xi,\theta\in\Theta)$. Moreover,
  $\sS(\qexp{T}\times\qexp{R})=\sS(\qexp{T})\times \sS(\qexp{R})$.
\end{definition}

\begin{definition}[Complete Quantum Statistical Models]
  A quantum statistical model
  $\qexp{T}=(\Xi,\mathcal{K},\boldsymbol{\tau})$ is said to be
  \emph{complete} if and only if $\sS(\qexp{T})$ is a complete state
  space.
\end{definition}

\begin{proposition}[Equivalent condition for sufficiency]\label{prop:blackwell-comp-pos}
  Given two quantum statistical models
  $\qexp{R}=(\Theta,\sH,\boldsymbol{\rho})$ and
  $\qexp{S}=(\Theta,\sH',\boldsymbol{\sigma})$, the following are
  equivalent:
  \begin{enumerate}
  \item \begin{equation}
    \qexp{R}\succ\qexp{S};
\end{equation}
\item
\begin{equation}
  \qexp{T}\times\qexp{R}\supset\qexp{T}\times\qexp{S},
\end{equation}
for every auxiliary quantum statistical model
$\qexp{T}=(\Xi,\mathcal{K},\boldsymbol{\tau})$;
\item
\begin{equation}
  \qexp{T}\times\qexp{R}\supset\qexp{T}\times\qexp{S},
\end{equation}
for some complete quantum statistical model
$\qexp{T}=(\Xi,\sK,\boldsymbol{\tau})$ with $\sK\cong\sH'$.
\end{enumerate}
\end{proposition}

\begin{proof}
  The implications ``1 $\Rightarrow$ 2'' and ``2 $\Rightarrow$ 3'' are
  trivial. In order to prove the implication ``3 $\Rightarrow$ 1'',
  let us consider an auxiliary quantum statistical model
  $\qexp{T}=(\Xi,\sH',\boldsymbol{\tau})$, such that $\sS(\qexp{T})$
  is complete for $\bound(\sH')$, according to
  Definition~\ref{def:complete}. The condition
  $\qexp{T}\times\qexp{R}\supset\qexp{T}\times\qexp{S}$ implies, by
  Theorem~\ref{coro:weak-qexp}, the existence of a statistical
  morphism $\mL:\sS(\qexp{T}\times\qexp{R})\to
  \sS(\qexp{T}\times\qexp{S})$ such that
  $\mL(\tau_\xi\otimes\rho_\theta)=\tau_\xi\otimes\sigma_\theta$, for
  all $\xi\in\Xi$ and all $\theta\in\Theta$. By the completeness of
  $\sS(\qexp{T})$, this implies that the linear map
  $\mL:\bound(\sH')\otimes\bound(\sH)\to
  \bound(\sH')\otimes\bound(\sH')$ must in fact have the form
  $\id\otimes\mL'$. We are hence in the position to apply
  Proposition~\ref{prop:ent_ext}, which proves the existence of a
  completely positive, trace-preserving map
  $\mE:\bound(\sH)\to\bound(\sH')$ such that
  $\sigma_\theta=\mE(\rho_\theta)$, for all $\theta\in\Theta$,
  i.e. $\qexp{R}\succ\qexp{S}$.
\end{proof}

The corollary above makes it apparent that complete positivity is
always related with the possibility of \emph{extending} a quantum
system (in this case, a quantum statistical model) by composing it
with an auxiliary one.

\section{Sufficiency of quantum information structures, without
  entanglement}\label{sec:shmaya}

We begin this section with the following definition:

\begin{definition}[Composition of Quantum Information
  Structures]\label{def:comp-str} Given two quantum
  information structures $\boldsymbol{\varrhoup}_{AB}=(\sH_A,\sH_B,\rho_{AB})$ and
  $\boldsymbol{\upomega}_{XY}=(\sH_X,\sH_Y,\omega_{XY})$, the composition
  $\boldsymbol{\varrhoup}_{AB}\otimes \boldsymbol{\upomega}_{XY}$ is defined as the triple
  $(\sH_{A}\otimes\sH_X,\sH_B\otimes\sH_Y,\rho_{AB}\otimes\omega_{XY})$.
\end{definition}

\begin{remark}
  From Definitions~\ref{def:localss},~\ref{def:compos},
  and~\ref{def:comp-str}, it simply follows that
\begin{equation}\label{eq:com-st-sp-cont}
  \sS_{BY}(\boldsymbol{\varrhoup}_{AB}\otimes \boldsymbol{\upomega}_{XY})\supseteq\sS_B(\boldsymbol{\varrhoup}_{AB})\times\sS_Y(\boldsymbol{\upomega}_{XY}).
\end{equation}
\end{remark}


\begin{definition}[Complete Information Structures]\label{def:coplete-info}
  A quantum information structure
  $\boldsymbol{\upomega}_{XY}=(\sH_X,\sH_Y,\omega_{XY})$ is
  \emph{complete} if and only if:
  \begin{enumerate}
  \item the local state space $\sS_{Y}(\boldsymbol{\upomega}_{XY})$ is
    complete (see Definition~\ref{def:complete}), and,
\item for any given linear map $\mL_Y:\bound(\sH_Y)\to\bound(\sH_Y)$,
  $(\id_X\otimes\mL_Y)(\omega_{XY})=\omega_{XY}$ if and only if
  $\mL_Y=\id_Y$.
\end{enumerate}
\end{definition}

\begin{remark}
  In order to explicitly show the existence of a complete information
  structure $\boldsymbol{\upomega}_{XY}=(\sH_X,\sH_Y,\omega_{XY})$,
  let us consider the family of information structures
  $\boldsymbol{\upomega}^p_{XY}=(\sH_X,\sH_Y,\omega_{XY}^p)$, for
  $p\in[0,1]$, where $\dim\sH_X=\dim\sH_Y=d$ and $\omega_{XY}^p$ is an
  isotropic state, that is,
\begin{equation}\label{eq:iso}
  \omega_{XY}^p:=p\Psi^+_{XY}+(1-p)\frac{\openone_{XY}}{d^2},
\end{equation}
with $\Psi^+_{XY}$ denoting a maximally entangled state in
$\sH_X\otimes\sH_Y$. These states are known to satisfy the second
condition in Definition~\ref{def:coplete-info} for $p\neq
0$~\cite{faithful}. Moreover, a simple calculation shows that
\begin{equation}
  \sS_{Y}(\boldsymbol{\upomega}_{XY}^p)=\left\{\left.p\sigma_{Y}+(1-p)\frac{\openone_{Y}}d\right|\sigma_{Y}\in\sS(\sH_{Y})\right\},
\end{equation}
meaning that, for $p\neq 0$, $\sS_{Y}(\boldsymbol{\upomega}_{XY}^p)$
is complete.
\end{remark}

We are now able to state the following:
\begin{proposition}[Comparison of quantum information
  structures]\label{prop:shmaya2}
  Given two quantum information structures
  $\boldsymbol{\varrhoup}_{AB}=(\sH_A,\sH_B,\rho_{AB})$ and
  $\boldsymbol{\varsigmaup}_{AB'}=(\sH_A,\sH_{B'},\sigma_{AB'})$, the
  following are equivalent:
\begin{enumerate}
\item
  \begin{equation}\label{eq:shmaya1}
    \boldsymbol{\varrhoup}_{AB}\succ\boldsymbol{\varsigmaup}_{AB'};
\end{equation}
\item
\begin{equation}
  \left[\boldsymbol{\upomega}_{XY}\otimes\boldsymbol{\varrhoup}_{AB}\right]\ \supset_{XA}\ \left[\boldsymbol{\upomega}_{XY}\otimes\boldsymbol{\varsigmaup}_{AB'}\right],
\end{equation}
for every auxiliary quantum information structure
$\boldsymbol{\upomega}_{XY}=(\sH_X,\sH_Y,\omega_{XY})$;
\item
\begin{equation}
  \left[\boldsymbol{\uppsi}_{XY}^+\otimes\boldsymbol{\varrhoup}_{AB}\right]\ \supset_{XA}\ \left[\boldsymbol{\uppsi}_{XY}^+\otimes\boldsymbol{\varsigmaup}_{AB'}\right],
\end{equation}
for some auxiliary quantum information structure
$\boldsymbol{\uppsi}_{XY}^+=(\sH_X,\sH_Y,\Psi^+_{XY})$, such that
$\Psi^+_{XY}$ is a maximally entangled pure state and
$\sH_X\cong\sH_Y\cong\sH_{B'}$;
\item
  \begin{equation}\label{eq:shmaya2}
  \left[\boldsymbol{\upomega}_{XY}\otimes\boldsymbol{\varrhoup}_{AB}\right]\ \supset_{XA}\ \left[\boldsymbol{\upomega}_{XY}\otimes\boldsymbol{\varsigmaup}_{AB'}\right],
\end{equation}
for some auxiliary complete quantum information structure
$\boldsymbol{\upomega}_{XY}=(\sH_X,\sH_Y,\omega_{XY})$ with
$\sH_Y\cong\sH_{B'}$.
\end{enumerate}
\end{proposition}

\begin{proof}
  The implications ``1 $\Rightarrow$ 2'' and ``2 $\Rightarrow$ 3'' are
  trivial. The implication ``3 $\Rightarrow$ 4'' follows from the fact
  that, from Eq.~(\ref{eq:iso}), any maximally entangled information
  structure is, in particular, complete. We hence prove only the
  implications ``4 $\Rightarrow$ 1''.

  Starting from~(\ref{eq:shmaya2}), Theorem~\ref{thm:main} guarantees
  the existence of a statistical morphism
  $\mL_{YB}:\sS_{YB}(\boldsymbol{\upomega}_{XY}\otimes\boldsymbol{\varrhoup}_{AB})\to
  \sS_{YB'}(\boldsymbol{\upomega}_{XY}\otimes\boldsymbol{\varsigmaup}_{AB'})$
  such that
  \begin{equation}\label{eq:faith}
    \omega_{XY}\otimes    \sigma_{AB'}=(\id_{XA}\otimes\mL_{YB})(\omega_{XY}\otimes\rho_{AB}).
  \end{equation}
  Since $\omega_{XY}$ is a complete state, Eq.~(\ref{eq:faith})
  implies that the linear map $\mL_{YB}$ must in fact have the form
\begin{equation}
  \mL_{YB}\equiv\id_Y\otimes\mL_B.
\end{equation}
Further, the fact that $\id_Y\otimes\mL_B$ is a statistical morphism from
$\sS_{YB}(\boldsymbol{\upomega}_{XY}\otimes\boldsymbol{\varrhoup}_{AB})$
to
$\sS_{YB'}(\boldsymbol{\upomega}_{XY}\otimes\boldsymbol{\varsigmaup}_{AB'})$
implies that $\id_Y\otimes\mL_B$ is also a statistical morphism, in particular,
from
$\sS_{Y}(\boldsymbol{\upomega}_{XY})\times\sS_{B}(\boldsymbol{\varrhoup}_{AB})$
to
$\sS_{Y}(\boldsymbol{\upomega}_{XY})\times\sS_{B'}(\boldsymbol{\varsigmaup}_{AB'})$,
because of Eq.~(\ref{eq:com-st-sp-cont}). Finally, since we assumed
that $\sS_{Y}(\boldsymbol{\upomega}_{XY})$ is a complete state space,
we can apply Proposition~\ref{prop:ent_ext} to show that, indeed,
$\boldsymbol{\varrhoup}_{AB}\succ\boldsymbol{\varsigmaup}_{AB'}$.
\end{proof}

\begin{remark}
  In Ref.~\cite{shmaya}, the statement ``3 $\Leftrightarrow$ 1'' is
  proved. Proposition~\ref{prop:shmaya2} shows that the hypotheses can
  in fact be relaxed so that only the property of \emph{completeness},
  rather than entanglement, is required. Let us consider, as an
  example, the set of isotropic states defined in~(\ref{eq:iso}). Such
  states are known to be separable for $p\le\frac 1{d+1}$. Hence, by
  fixing a value $p_*\in\left(0,\frac 1{d+1}\right]$, we have that
  $\omega_{XY}^{p_*}$ is complete, induces a complete state space on
  $Y$, and, yet, it is a separable state. This fact recalls the
  results of Ref.~\cite{faithful}, where it was first noted how
  completeness (there referred to as ``faithfulness'') can replace
  entanglement, although in a different contest (namely, quantum process
  tomography).
\end{remark}

\begin{remark}
  In Remark~\ref{rem:from-str-to-qmod} we described how quantum
  statistical models can be identified with those quantum information
  structures, for which a decomposition like that in
  Eq.~(\ref{eq:bip-exp}) exists. One should hence expect that
  Proposition~\ref{prop:shmaya2} implies
  Proposition~\ref{prop:blackwell-comp-pos}, whenever
  $\boldsymbol{\varrhoup}_{AB}=(\sH_A,\sH_B,\rho_{AB})$ and
  $\boldsymbol{\varsigmaup}_{AB'}=(\sH_A,\sH_{B'},\sigma_{AB'})$ can
  be written in the form of Eq.~(\ref{eq:bip-exp}). In such a case,
  indeed, the fourth statement of Proposition~\ref{prop:shmaya2} can
  be used to re-derive Proposition~\ref{prop:blackwell-comp-pos}
  simply by considering an auxiliary quantum information structure
  $\boldsymbol{\upomega}_{XY}=(\sH_X,\sH_Y,\omega_{XY})$ of the form
  \begin{equation}
    \omega_{XY}:=\frac 1{|\Xi|}\sum_{\xi\in\Xi}|\xi\>\<\xi|_X\otimes\tau^\xi_Y.
\end{equation}
The crucial observation is that the above quantum information
structure is complete if and only if the corresponding quantum
statistical model
$\qexp{T}_{\boldsymbol{\upomega}}:=(\Xi,\sH_Y,\boldsymbol{\tau})$,
with $\boldsymbol{\tau}=(\tau_\xi;\xi\in\Xi)$, is complete. The rest
of the proof is left to the interested reader.
\end{remark}

\section{Conclusions}\label{sec:concl}

We extended some results from the theory of comparison of statistical
models to quantum statistical decision theory. This has been done by
relaxing Petz's definition of coarse-grainings to that of statistical
morphisms. By using such generalized notion, we introduced comparison
criteria for quantum statistical models and quantum information
structures, which are the direct generalization to a non-commutative
setting of the comparison criteria used in classical decision
theory. The framework we described turned out to be general enough to
encompass both the classical and the quantum case. We showed how
results that previously were independent, like the
Blackwell-Sherman-Stein theorem for statistical models and Shmaya's
result for quantum information structures, can be in fact recovered as
special cases of a single, unifying comparison theorem, which also
sheds new light on both: the BSS theorem has been extended to a
quantum-classical scenario, and Shmaya's comparison criterion has been
strengthened by removing the need of auxiliary entangled resources.

As a final remark, the reader might have noticed that, as long as the
states of a statistical theory can be represented by self-adjoint
matrices (not necessarily positive) of unit trace, the definitions of
information ordering and m-sufficiency proposed here can be
straightforwardly extended to consider such cases as well. For such
\emph{generalized probabilistic theories}, an extension of the BSS
theorem can also be proved, along the same lines described in the
present work.

\section*{Acknowledgements}
The author is grateful to Masanao Ozawa for illuminating conversations
and clarifying suggestions. Discussions with Giacomo Mauro D'Ariano
and Madalin Guta greatly contributed in improving the presentation of
the work. This research was supported by the Program for Improvement
of Research Environment for Young Researchers from Special
Coordination Funds for Promoting Science and Technology (SCF)
commissioned by the Ministry of Education, Culture, Sports, Science
and Technology (MEXT) of Japan.  Part of this work has been done when
the author was visiting the Statistical Laboratory of the University
of Cambridge.

\appendix

\newpage\section{Proof of Proposition~\ref{prop0}}\label{app:a}

\begin{proposition0}
  For any two given statistical models
  $\exp{E}=(\Theta,\Delta,\boldsymbol{\alpha})$ and
  $\exp{F}=(\Theta,\Delta',\boldsymbol{\beta})$, $\exp{E}\supset
  \exp{F}$ if and only if $\exp{E}\supset_{\mathrm{Bayes}}\exp{F}$.
\end{proposition0}

\begin{proof}
  The statement can be proved by using the Separation Theorem between
  convex sets~\cite{rocka} as follows. (Notice that in our case all
  convex sets are closed and bounded, so that we can proceed without
  paying attention to too many technical details.)

  Generally speaking, the convex set
  $\mathcal{C}_1\subset\mathds{R}^N$ is not contained in the convex
  set $\mathcal{C}_2\subset\mathds{R}^N$ if and only if there exists a
  point $\vec{v}\in\mathcal{C}_1$ such that
  $\vec{v}\notin\mathcal{C}_2$. Then, the Separation Theorem
  (Corollary 11.4.2 of Ref.~\cite{rocka}), applied to the convex set
  $\mathcal{C}_2$ and the single-point (hence convex) set
  $\{\vec{v}\}$, states that, for such $\vec{v}$, there exists a
  vector $\vec{b}\in \mathds{R}^N$ such that
\begin{equation}
  \max_{\vec{w}\in\mathcal{C}_2}\sum_{n=1}^Nb^nw^n<
  \sum_{n=1}^Nb^nv^n.
\end{equation}
Equivalently, we can say that the convex set
$\mathcal{C}_1\subset\mathds{R}^N$ is contained in the convex set
$\mathcal{C}_2\subset\mathds{R}^N$ if and only if, for all vectors
$\vec{b}\in \mathds{R}^N$,
\begin{equation}
\max_{\vec{w}\in\mathcal{C}_2}\sum_{n=1}^Nb^nw^n\ge
\max_{\vec{v}\in\mathcal{C}_1}\sum_{n=1}^Nb^nv^n.
\end{equation}

Moreover, for any given non-vanishing probability distribution
$\pi(n)$, $\sum_n\pi(n)=1$, the convex set
$\mathcal{C}_1\subset\mathds{R}^N$ is contained in the convex set
$\mathcal{C}_2\subset\mathds{R}^N$ if and only if, for all vectors
$\vec{b}\in \mathds{R}^N$,
\begin{equation}
  \max_{\vec{w}\in\mathcal{C}_2}\sum_{n=1}^N\pi(n)b^nw^n\ge
  \max_{\vec{v}\in\mathcal{C}_1}\sum_{n=1}^N\pi(n)b^nv^n.
\end{equation}
This follows from the fact that the above equation has to hold for all
$\vec{b}\in \mathds{R}^N$, so that the non-vanishing probabilities
$\pi(n)$ can be absorbed in the definition of $\vec{b}$. In
particular, there is no loss of generality in considering
$\pi(n)=1/N$, for all $n$.

We now turn to the case of $\mathcal{C}(\exp{E},\outx,\ell)$ and
$\mathcal{C}(\exp{F},\outx,\ell)$, choosing the \emph{a priori}
probability on $\Theta$ as $\pi(\theta)=1/|\Theta|$, for all
$\theta$. Then, for every $\vec{b}\in\mathds{R}^{|\Theta|}$,
\begin{equation}
  \max_{\phi:\textrm{
      $\outx$-r.d.f.}}\frac 1{|\Theta|}\sum_{\theta\in\Theta}b^\theta v^\theta(\phi;\exp{E},\outx,\ell)=\max_{\phi:\textrm{
      $\outx$-r.d.f.}}\frac 1{|\Theta|}\sum_{\theta\in\Theta}v^\theta(\phi;\exp{E},\outx,\widetilde{\ell}),
\end{equation}
where the function $\widetilde{\ell}$ at the left hand side is another
payoff function with such that
$\widetilde{\ell}(\theta,i)=\ell(\theta,i)b^{\theta}$. In other words,
the vector $\vec{b}$ can be absorbed in the definition of the payoff
function. This means that, for any finite set of decisions $\outx$ and
any payoff function $\ell:\Theta\times\outx\to\mathbb{R}$,
$\mathcal{C}(\exp{E},\outx,\ell)\supseteq
\mathcal{C}(\exp{F},\outx,\ell)$ if and only if, for every payoff
function $\widetilde{\ell}:\Theta\times\outx\to\mathbb{R}$,
\begin{equation}\label{eq:call_sep}
  \max_{\phi:\textrm{
      $\outx$-r.d.f.}}\frac 1{|\Theta|}\sum_{\theta\in\Theta}v^\theta(\phi;\exp{E},\outx,\widetilde{\ell})\ge
  \max_{\phi':\textrm{ $\outx$-r.d.f.}}\frac 1{|\Theta|}\sum_{\theta\in\Theta}v^\theta(\phi';\exp{F},\outx,\widetilde{\ell}),
\end{equation}
where the maxima are taken over all possible $\outx$-r.d.f. $\phi$ on
$\Delta$ and $\phi'$ on $\Delta'$. This, in turns, implies the
statement.
\end{proof}

\end{document}